\documentclass[pdflatex,sn-mathphys-num]{sn-jnl}

\usepackage{amsmath}
\usepackage{amsthm}
\usepackage{amssymb}
\usepackage{graphicx}
\usepackage{algpseudocode}
\usepackage{algorithm}
\usepackage{enumerate}
\usepackage{lipsum}

\theoremstyle{thmstyleone}%
\newtheorem{theorem}{Theorem}%  meant for continuous numbers
\theoremstyle{thmstyletwo}%

\theoremstyle{thmstylethree}%
\newtheorem{definition}{Definition}%

\newtheorem{lemma}{Lemma}{}
{}
{}

\newcommand{\cost}{\textsc{Cost}}

\newcommand{\OPT}{\text{\sc{OPT}}}
\newcommand{\RR}{\mathbb{R}}
\newcommand{\eeps}{\tilde{O}(1)}

\newcommand{\eps}{\varepsilon}

\newcommand{\expected}{\mathbb{E}}

\newcommand{\comp}{\textsc{Comp}}

\newcommand{\ses}{\sigma}
\newcommand{\TT}{\mathcal{T}}
\newcommand{\SSS}{\mathcal{S}}
\newcommand{\PP}{\mathcal{P}}

\newcommand{\npairs}{n_{\text{pairs}}}

\newcommand{\nport}{n_{\text{port}}}
\newcommand{\flowgraph}{\textsc{Flow graph}}
\newcommand{\flow}{\textsc{Flow}}
\newcommand{\flowgraphs}{\textsc{Flow graphs}}
\newcommand{\FF}{\mathcal{F}}
\newcommand{\DP}{\text{DP}}
\newcommand{\correct}{\text{Upgrade}}
\newcommand{\newGamma}{\mathcal{T}}
\newcommand{\muu}{\mu}

\newcommand{\dellta}{\delta} %CHECK THESE??
\newcommand{\deltta}{\delta} %CHECK THESE??
\newcommand{\deltaa}{\delta} %CHECK THESE??

\begin{document}

\title{An $n^{O(\log\log n)}$ time approximation scheme for capacitated VRP in the Euclidean plane.} 

\author{\fnm{Ren\'e} \sur{Sitters}}\email{r.a.sitters@vu.nl}

\affil{\orgname{Vrije Universiteit Amsterdam}}

\abstract{
We present a quasi polynomial time approximation scheme (Q-PTAS) for the capacitated vehicle routing problem (CVRP) on $n$ points in the Euclidean plane for arbitrary capacity $c$. The running time is $n^{f(\epsilon)\cdot\log\log n}$ for any $c$, where $f$ is a function of $\epsilon$ only. This is a major improvement over the so far best known running time of 	$O(n^{\log^6 n/\epsilon^5})$ and a big step towards a PTAS for Euclidean CVRP. 

In our algorithm, we first give a polynomial time  reduction of the CVRP in $\RR^d$ (for any fixed $d$) to an uncapacitated routing problem in $\RR^d$ that we call the $m$-paths problem. Here, one needs to find exactly $m$ paths between two points $a$ and $b$, covering all the given points in the Euclidean space. We then give a Q-PTAS for the $m$-paths problem in the plane. Any PTAS for the (arguably easier to handle) Euclidean $m$-paths problem is most likely to imply a PTAS for the Euclidean CVRP. \\

}
		

		\maketitle
		
		\section{Introduction} 
		The Vehicle Routing Problem (VRP), introduced by Dantzig and Ramser in 1959~\cite{DR59}, is a fundamental  optimization problem in operations research and theoretical computer science, and has been studied in many variations. 
		Here, we give an algorithm for the classic Capacitated Vehicle Routing Problem (CVRP)~\cite{DR59}. 
		An instance consists of a capacity $c$, and a set of points in a metric space, where one point is the origin (or depot), and the others are the locations to be visited. A feasible solution is a set of tours that together cover all points, and every tour contains the origin plus at most $c$ locations. The objective is to minimize the total distance. We study the version where all points are in the Euclidean plane. The CVRP is also known as the $k$-tours problem or the cycle covering problem~\cite{AsanoKTT97,Arora98JACM,blauth2023improving}.

		Haimovich and Rinnooy Kan~\cite{Haimovich1985bounds} where the first to give a polynomial time approximation scheme (PTAS) for CVRP in the Euclidean plane. The running time of their algorithm is exponential in $c$, but still polynomial for $c=O(\log\log n)$. Asano et al.~\cite{AdamaszekCL2010} increase this to $c=O(\log n/\log\log n)$ and also show that any PTAS for TSP in the Euclidean plane~\cite{Arora98JACM,Mitchell1999} implies a PTAS for CVRP with $c=\Omega(n)$.
		Adamaszek et al.~\cite{AdamaszekCL2010} increased the capacity for which a PTAS exists to  $c=2^{O(\log^{\dellta}n)}$, for some $\dellta$ depending on $\epsilon$. In their algorithm, they first reduce the number of points substantially by partitioning the instance and moving points to grid points, and then apply  the \emph{quasi} polynomial time approximation scheme (Q-PTAS) by Das and Mathieu~\cite{DasMathieu2014} as a black box. 
		The Q-PTAS from~\cite{DasMathieu2014} builds on the PTAS for Euclidean TSP by Arora~\cite{Arora98JACM}, and the quasi polynomial time is achieved by rounding the number of points within the dynamic program. The DP-solution then is a set of tours that may violate the capacity constraint by a small factor. This is resolved by removing points from tours, which are then assigned to new tours. 
		Jayaprakash and Salavatipour~\cite{JayaprakashS2023} improved the running time of the Q-PTAS to $O(n^{\log^6 n/\epsilon^5})$.
		For higher dimensional Euclidean metrics, Khachay and Dubinin [17] gave a PTAS for fixed dimension $d$ and $c = O(\log^{1/\epsilon}n)$.
		   
		For arbitrary metrics, Haimovich and Rinnooy Kan~\cite{Haimovich1985bounds} gave a $(1 + \alpha)$-approximation
		where $\alpha$ is the approximation factor for TSP. This \emph{iterated tour partitioning} algorithm  first solves the TSP on the input points and then partitions the tour into segments of at most $c$ points each. 
		Recently, Blauth et al.\cite{blauth2023improving} improved this ratio to $1+\alpha-\deltta$, for some small constant $\deltta>0$. For arbitrary metrics, the CVRP is known to be APX-hard even for constant capacity $c\ge 3$~\cite{AsanoKTT97,MiltenburOS2024}. 
		
		Capacitated vehicle routing has been widely studied for tree metrics as well. The CVRP on trees is NP hard, which follows by an easy reduction from bin packing~\cite{labbe1991capacitated}. Hamaguchi and Katoh~\cite{HamaguchiK1998} designed a $1.5$-approximation, which was improved to $(\sqrt{41}-1)/4$ by Asano et al.~\cite{AsanoKK2001}, and to $4/3$ by Becker~\cite{Becker18}. Jayaprakash and Salavatipour~\cite{JayaprakashS2023} gave a Q-PTAS before finally, Mathieu and Zhou~\cite{mathieu2023ptastrees} designed a PTAS for CVRP on trees.
		
		In the unsplittable demand setting of CVRP, each point is given a demand that cannot be split over different tours and the total demand served by each tour should stay within the capacity $c$. The iterated tour partitioning algorithm by Haimovich and Rinnooy Kan~\cite{Haimovich1985bounds} yields a $(2 + \alpha)$-approximation, where $\alpha$ is the approximation factor for TSP. 
		This was recently improved to $2+\alpha-\deltaa$, for some small constant $\deltaa>0$~\cite{blauth2023improving}.  For the Euclidean plane, hardness of bin packing implies a $3/2$ lower bound on the approximation and the current best approximation ratio of $2+\epsilon$ for unsplittable CVRP was given by Grandoni et al.~\cite{GrandoniMZ2022unsplittablR2}.

		\paragraph*{Model and results}
		An instance of the Euclidean CVRP is given by an integer $c\ge 1$ and $n$ points (the input points) plus one point $r$ called the origin (or depot). All points are in the Euclidean metric space with distance function $d(,)$.  A feasible solution is a set of tours that together cover all points and every tour goes through the origin and contains at most $c$ other points. The objective is to minimize the total distance.
		We shall often refer to the points to be visited as the \emph{input points}  to distinguish them from points defined by the algorithm such as `portals', `anchor points', or `points of intersection'.        	
		
	 We use the notation $\tilde{O}()$ to hide dependency of the hidden constant on $\epsilon$ in the big-O notation. Throughout the paper we assume w.l.o.g. that $0<\eps<1$ and $1/\epsilon$ is integer. By default, the base of the logarithms is 2, i.e.,   	
		$\log n:=\log_2 n$.		
		\begin{theorem}
			There is an approximation scheme for CVRP in the Euclidean plane with running time $n^{\tilde{O}(\log\log n)}$, which holds for any capacity $c$.
		\end{theorem}
		
		The running time of our Q-PTAS is a major improvement over the 
		$n^{\log^{O(1/\epsilon)}n}$ time of the Q-PTAS by Das and Mathieu~\cite{DasMathieu2014} and the 
		$O(n^{\log^6 n/\epsilon^5})$ time Q-PTAS from~\cite{JayaprakashS2023}, 
		 and a big step towards a PTAS for Euclidean CVRP. Moreover, we provide a way to take the last hurdle by giving a polynomial time reduction from CVRP to what we call the \emph{$m$-paths problem}, which is, arguably, an easier problem to handle than the CVRP. Any PTAS for the Euclidean $m$-paths problem is most likely to imply a PTAS for the Euclidean CVRP.

		\begin{figure}
			\centering
			\includegraphics[width=0.40\linewidth]{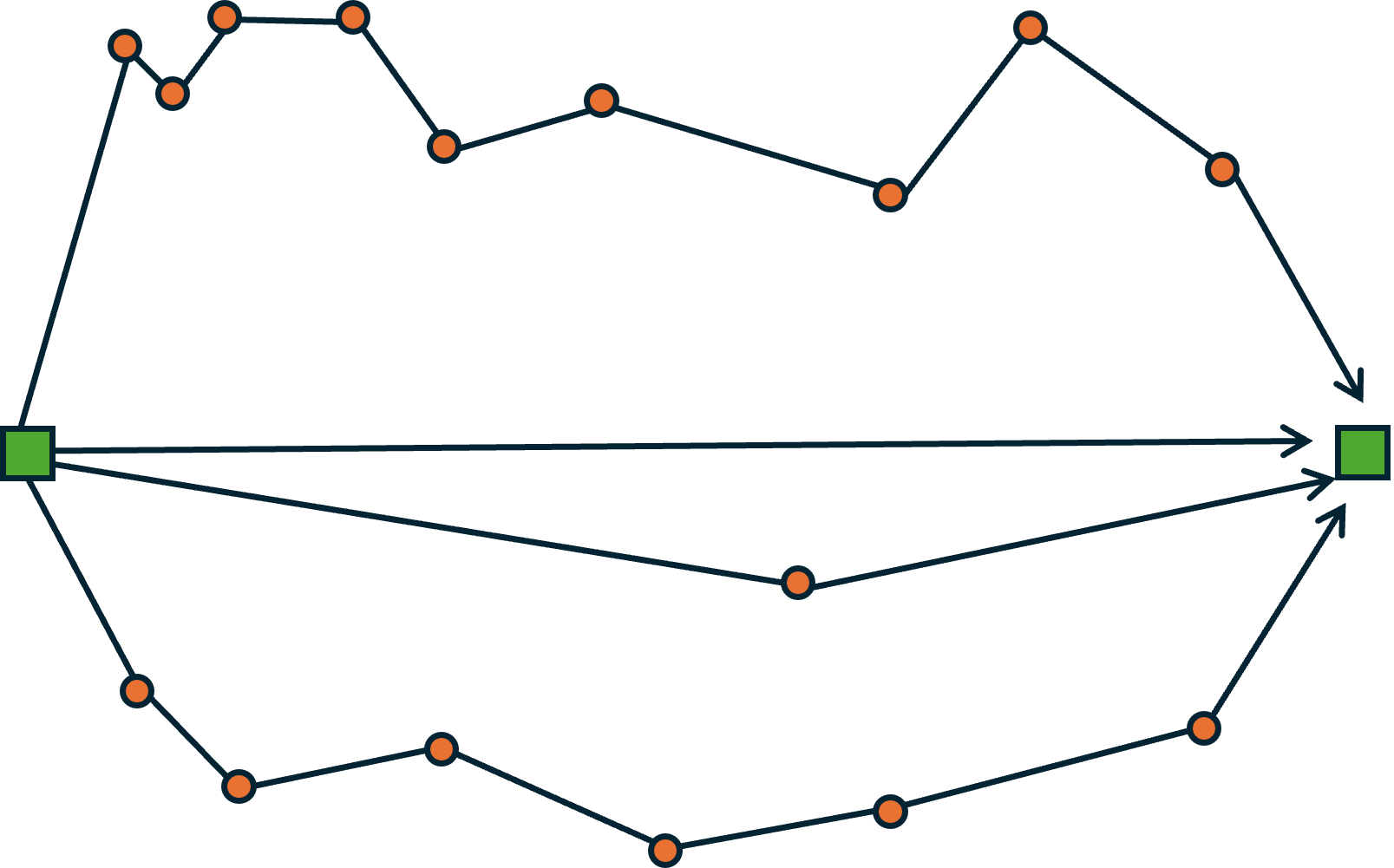}
			\caption{In the Euclidean $m$-paths problem, we need to find exactly $m$ paths between two given points $a$ and $b$, that together cover all the input points. Goal is to minimize total length. The example is for $m=4$. The TSP-path problem with given end points corresponds to the $m$-paths problem for $m=1$.
				}
			\label{fig:mpaths}
		\end{figure}
		
Our algorithm, when implemented as described, has an absolutely impractical running time. The number  $\log_2\log_2 n$ in the exponent is small for any meaningful value of $n$ (for example, $\log_2\log_2(10^{10})<6$) but this is misleading since the exponent is dominated by the constant hidden in the $\tilde{O}$ notation. In fact, for  $\log_2\log_2 n$ to be larger than the hidden constant, the value $n$ must be  \emph{beyond astronomically large}!    

Our algorithm uses a wide variety of techniques, some easy and standard, others sophisticated and less common. 
Two new ingredients are of particular interest. In Section~\ref{sec:loglog}, a new portal layout is defined (See Figure~\ref{fig:portallayout}) that provides a particular lower bound on any portal respecting solution, which does not follow from the standard portal layout~\cite{Arora98JACM}. In the same section, we give a sophisticated way of rounding the number of paths in the dynamic program for the $m$-paths problem. Although this form of rounding is a natural step at this point in the algorithm, it is quite challenging since a DP-solution is no longer a set of tours (or paths). 
The difficulty lies in choosing the structure of the subproblems in the DP in such a way that the obtained pseudo-solution  can be repaired at small cost. 
In the next section we give a high-level overview of the proof before going over all the details in the succeeding sections.

		\subsection{Outline of the algorithm and proof}
		The Q-PTAS contains numerous details but the main ideas are rather simple. A high level overview is displayed in Algorithms 1,2, and 3 below. 
		First, we apply an established technique (Algorithm 1) in which we reduce the problem to \emph{bounded instances}. A set of instances is called bounded if there is some constant number $Q$ such that the ration of maximum over minim distance to the origin, $d_{\max}/d_{\min}$, is at most $Q$ for every instance in the set.   
		
		For a bounded instance we find a solution by dividing the space into $\tilde{O}(1)$ squares of equal size  and solve many subproblems on each square. Solutions for subproblems are merged into a feasible solution for the CVRP. We do not use a hierarchical DP (as in Euclidean TSP) but basically apply  complete enumeration. The solution space is restricted by defining $\eeps$ points on the boundary of each square and restrict any crossing of a tour with the boundary of a square to these points, which we call \emph{anchor points}. 
		In total, there are only $\eeps$ anchor points and we say that two tours are of the same \emph{type} if their sequence of anchor points is the same. 
		We refer to the part of a tour that falls inside some square as a \emph{segments}. Ignoring capacity constraints, two tours of the same type can swap their segment in any square. (See Figures~\ref{fig:same_type} and~\ref{fig:swapsegments}.) 
		We show that if we have a large number $k$ of tours of the same type such that their total number of points is at most $kc$, then by swapping we can reassign points to tours such that each tour has at most $c$ points, approximately. Fixing the bound to exactly $c$ is done in a greedy way.    
		This key observation in our proof is rather similar to bin packing. 
		We use standard bin packing techniques, such as rounding the item sizes (here, the number of points on a segment), and solving a linear program to assign items to bins (here segments to tours), and then adding the fractional items (fractional segments) to tours in a greedy way. (See Algorithm 2.)
		This bin packing approach works for bounded instances since we can choose the grid dense enough such that the size of a single grid cell is small compared to the length of any tour.   
		
		An instance of any subproblem on a square is given by the boundary information (how the tours enter and leave the square) and for each tour type, the number of points visited by all tours of that type together. This 'together' is the key to the efficiency of the Q-PTAS. We do not need to keep track of the number of points on each tour, but only on the total for each type and we have only $\eeps$ different types of tours. 
		So we basically got rid of the complexity due to the capacity constraints since there are only $\eeps$ of them. But we still need to deal with a non-constant number of tours in each square. We refer to the subproblem that needs to be solved as the \emph{general $m$-paths problem}. 
		The easiest version of this subproblem, which we call the \emph{$m$-paths problem}, is defined as follows. We are given a square containing a set of points, two points $a$ and $b$ on the boundary, and a number $m$. Goal is to find exactly $m$ paths between $a$ and $b$ that together visit all points in the square. (See Figure~\ref{fig:mpaths}.) Note that the case $m=1$ is similar to the TSP-path problem. In fact, for constant $m$, a PTAS follows easily using Arora's approach~\cite{Arora98JACM}. But $m$ is not constant in general and keeping track of all the different paths in the DP is non-trivial. To simplify analysis we assume tours are directed. Then, the $m$ paths problem is to construct a directed graph with flow values on the arcs defining an integral  flow from $a$ to $b$ of value $m$ and such that the support graph is connected and contains every input point. An important observation is that when we solve the $m$-paths problem using Arora's approach in a pretty straight forward way, we do get an $n^{\tilde{O}(\log n)}$ time approximation scheme if, in the DP, we store precise information on the flow values, i.e., the number of paths crossing each portal in either direction. The running time  follows easily from the fact that there are $\tilde{O}(\log n)$ portals per dissection square and each portal is crossed $O(m)=O(n)$ times. We describe the details of this approach in this paper for three reasons: 
		First, this is already a big improvement over the running time of 
		$n^{\log^{O(1/\epsilon)}n}$ from~\cite{DasMathieu2014} but also over the 
		$O(n^{\log^6 n/\epsilon^5})$ time from~\cite{JayaprakashS2023}. Second, the DP and its analysis stay close to that of Arora~\cite{Arora98JACM} and it is relatively easy to verify the correctness of this part of the algorithm.  Third, it forms the basis for the more complex analysis of the $n^{\tilde{O}(\log\log n)}$ time algorithm. 
		
		To get the improved running time, we round flow values in the DP to powers of some number $\alpha$. Of course, using rounded numbers like this gives a mismatch of flow values at portals, i.e., we lose flow conservation. 
		But intuitively, modest rounding gives small errors that can be fixed at a small cost. Hence, we should get \emph{some} improvement over the $n^{\tilde{O}(\log n)}$ time algorithm of Section~\ref{sec:logn}, which does not use rounding. 
		The details of the Q-PTAS for the $m$-paths problem are rather involved as we need to carefully keep track  of the rounding error. We cannot simply apply the same DP with rounded values but need to add more structure to the configurations stored in the DP. For example, we use a sophisticated portal layout (see Figure~\ref{fig:portallayout}) which is more restricted than the standard layout, and provides a new lower bound on the cost of any portal respecting solution but gives no significant increase in the cost of the solution.  This layout is of general interest for TSP-like problems in the Euclidean plane.

		After analyzing the $m$-paths problem, we discuss the complexity of the \emph{general} $m$-paths problem, which is basically a combination of $\eeps$ $m$-paths problems that have to solved simultaneously on the same square. In the DP, these different $m$-paths problems are only dealt with together at the lowest level (the grid cells), where we assign input points to paths. At all other levels in the DP, we can handle the different $m$-paths problems independently in parallel, which leads to an $\eeps$ factor increase in the exponent of the running time. Hence,  the running time remains $n^{\tilde{O}(\log\log n)}$.

		\section{A polynomial time reduction from CVRP to the $m$-paths problem. }
		The reduction to the $m$-paths problem consists of two steps, as described in Algorithm 1 and 2. First, we reduce to $O(n)$ bounded instances and then reduce each bounded instance to  $n^{\eeps}$ instances of the $m$-paths problem.

		\begin{algorithm}
			\caption{Reduction to bounded instances:}\label{alg:main}
			\begin{algorithmic}
				\State - Partition the instance $I$ (at random) into bounded instances  $I_i$. 
				\State - Solve each bounded instance $I_i$, using \textbf{Algorithm 2}. This gives a set of
				\State\ \   tours $\newGamma_i$ for each $I_i$.  \\
				\Return  the union: $\cup_i \newGamma_i$. 
			\end{algorithmic}
			\label{alg:1}
		\end{algorithm}
		\begin{algorithm}
			\caption{Reduction from a bounded instance to the $m$-paths problem:}\label{alg:bounded}
			\begin{algorithmic}
				\State - Place a grid of size $\tilde{O}(1)$ over the input points. 
				\State -   For each square (grid cell), use \textbf{Algorithm 3} to solve $n^{\tilde{O}(1)}$ subproblems
				\State \ \  (the $m$-paths problems). 
				\For {every combination of one subproblem for each square}
				\State - Solve a linear program and obtain a partial solution.
				\State - Add the missing parts in a greedy way and store the solution. 
				\EndFor\\
				\Return the best solution found. 
			\end{algorithmic}
			\label{alg:2}
		\end{algorithm}\begin{algorithm}
			\caption{Algorithm for the $m$-paths problem (the subproblem):} \label{alg:subproblem}
			\begin{algorithmic}
				\State -  Apply a dynamic program similar to Arora using rounded numbers. 
				\State - Turn the pseudo solution returned by the DP into a real solution. \\
				\Return the solution.
			\end{algorithmic}
			\label{alg:3}
		\end{algorithm}

		\subsection{Reducing to bounded instances}\label{sec:decompositionTSP}
		
		Bounding the ratio of maximum and minimum distance to the origin is a first step in several approximation schemes. 
		Adamaszek et al.\cite{AdamaszekCL2010} (2010) introduce the technique to reduce the number of points in the instance to some function of $\epsilon$ and $c$ and then show that the Q-PTAS from~\cite{DasMathieu2014} has polynomial running time for moderately large capacities $c=2^{\log^{O(\epsilon)}n}$.  
		In their reduction to bounded instances, they use a Baker's type approach~\cite{Baker1994}, where the plane is partitioned into geometrically increasing rings that partition the point set into subsets $U_1,U_2,U_3,\dots$. Then, every $(1/\epsilon)$-th. ring is `marked'. Removing the marked rings leaves a collection of bounded instances (separated by the marked rings). The instances on the marked rings contribute (when randomized) only an $O(\epsilon)$ fraction of the optimal cost and are approximated using the simple iterated tour partitioning algorithm. Grandoni at al~\cite{GrandoniMZ2022unsplittablR2} adopt the same partitioning scheme, and Mathieu and Zhou~\cite{mathieu2023ptastrees,MathieuZhou2024Tight-1.5trees} use a similar approach for trees, splitting the instance in two types of bounded instances,  where one of the type contributes only $O(\epsilon)$ to the optimal cost and can be handled in a less efficient way. 
		
		A reduction to bounded instances has also been used for the average completion time objective~\cite{Sitters2021,GriesbachHKS_ESA_2023}. Our approach here is in fact based on the reduction used in the PTAS for the traveling repairman problem \cite{Sitters2021}. A difference of our approach with that used in the VRP literature is that we have only one type of bounded instances, and is, in our opinion, simpler and more general.  The same reduction was also used in~\cite{MiltenburgSOFSEM}. For completeness, we include the  relatively short proof in this section.
		Denote by $d(v)$ the distance $d(r,v)$. Let $d_{\max}=\min_v d_v$ and $d_{\min}=\min_v d_v$. Assume $d(v)\ge 1$ for all $v$.
		
		\bigskip\noindent \textbf{Reduction to bounded instances} (Algorithm 1):
		\label{alg:decompose} 
		\begin{itemize}
			\item[1)] \textbf{Partition:} Let $a=2/\eps$ and take $b$ uniformly at random in $[0,1]$.
			
			Let $x_0=0$ and $x_i=e^{a(i-1+b)}$ for $i=1,2,\dots,\muu$, where $\muu$ is large enough such that $d_{\max} < x_{\muu}$. Partition the points  into sets $V_i=\{v|x_{i-1}\le d(v)<x_{i}\}$, $i\in\{1,2,\dots,\muu\}$.
			\item[2)] \textbf{Approximate subproblems:} Let $I_i$ be the instance restricted to the points in $V_i$. For each $I_i$, find a solution $\newGamma_i$.
			
			\item[3)] \textbf{Concatenate:} Return the union  of $\newGamma_1,\dots,\newGamma_\muu$ as the final solution.
		\end{itemize}\bigskip

		The key observation in the reduction to bounded instances is that we choose $a$ large enough as a function of $\eps$ only such that the expected loss due to the decomposition is only an $\eps$ fraction of the optimal cost.   
		
		\begin{lemma}
		The union of $\alpha$-approximate solutions $\newGamma_i$ for the bounded instances $I_i$ gives a $(1+\epsilon)\alpha$-approximation for instance $I$.
		\end{lemma}
		\begin{proof}
			Consider an optimal solution $\newGamma^*$ for $I$. We shall randomly decompose $\newGamma^*$ into solutions $\newGamma_i'$ for each $V_i$ such that the expected increase in cost is no more than $\eps$ times the cost of $\newGamma^*$.  
		
		We say that an edge $(u,v)$ of a tour \emph{crosses a circle} at point $s$ if $s$ is on  $(u,v)$ and one endpoint lies inside the circle and the other outside or on the circle.  
		Suppose an edge of a tour $T$ crosses a circle at $s$. We open up  $T$ at $s$ and we add an edge $(r,s)$ for each of the two parts. We do this for every crossing of a tour with a circle
		(but not for the new crossings made by the added edges). 
		This resulting set of tours, say $\newGamma'$, is obviously a feasible solution and has the additional property that any tour in it only visits points from one subset $V_i$. Hence $\newGamma'$ is the sum of solutions $\newGamma_1', \newGamma_2', \dots,\newGamma_{\muu}'$, where each $\newGamma_i'$ is a solution for $V_i$.   
		
		Next, we compute the expected cost for the added edges.
		Let $x_i(b)=\alpha_i e^{ab}$ with $\alpha_i=e^{a(i-1)}$ and  $b$ uniformly in $[0,1]$, as defined above. 
		The function $x_i(b)$ has range $[e^{a(i-1)},e^{ai}]$ and is invertible on its domain $[0,1]$. 
		Hence, the probability density function $f_i$ for $x_i$ can easily be computed: 
		\[f_i(x_i))=\frac{1}{h_i'(b)}=\frac{1}{\alpha_i a e^{ab}}=\frac{1}{\alpha_i ae^{\ln(x_i/\alpha_i)}}
		=\frac{1}{ax_i}\text{ for }x_i\in [e^{(i-1)a},e^{ia} ].\]
	Let $(u,v)$ be an edge in a tour $T$ and assume $d(u)\le d(v)$. 
		Note that $(u,v)$ crosses a circle of radius $x$ if $d(u)<x<d(v)$. The total expected sum of the radii of circles cutting $(u,v)$ is  
		\[\int\limits_{x=d_u}^{x=d_v}x\cdot \frac{1}{ax}dx=\int\limits_{x=d_u}^{x=d_v} \frac{1}{a}dx=(d_v-d_u)/a\]
		For every crossing of an edge with a circle we pay twice the radius of the circle in our decomposition of $\OPT$  
		Hence, the total expected additional cost is no more than 
		$2\OPT/a = \eps \OPT$. Let $\OPT_i$ be the optimal cost for instance $V_i$. Note that this is random variable and 
		\[\sum_i \expected[ \OPT_i]= \expected[ \sum_i \OPT_i]\le (1+\eps)\OPT.\]
		\end{proof}
		
		\subsection{Defining the subproblems for a bounded instance.}	
		Consider a bounded instance and denote by $\OPT$ its optimal cost. (We drop the index $i$ that we used for the bounded instances in the previous section.) Let 
		$\max d_j\le D$ and $\min d_j\ge \delta D$ for some $D>0$ and $0<\delta<1$.
		Form the previous section we can assume that $\delta\le e^{2/\epsilon}$ but we will just work $\delta=\eeps$ in the analysis.
		We place a grid, in random position, that partitions the plane into squares (the grid cells). (See  Figure~\ref{fig:squares}.) To be precise, assume all input points are placed in a $2D\times 2D$ square $\mathcal{B}$ with the origin in the middle of the square at location $(D,D)$. We choose integer $q=\eeps$ appropriately later and define $\tau=2D/(q-1)$. We place a square box $\mathcal{B'}$ of side length $q\tau$ on top of $\mathcal{B}$. The box $\mathcal{B}'$ is divided into $q^2$ squares of side length $\tau$. The lower left corner of $\mathcal{B}'$ is placed in the point $(-x,-y)$ where $x$ and $y$ are taken uniformly at random from the interval $[0,\tau]$. 
		Hence, $\mathcal{B'}$ will always cover $\mathcal{B}$.   
		
		On each side of each of the $q^2$ squares we place $1/\eps$ points at distance $\eps\tau$. We call these \emph{anchor points}. 
		We do not put anchor points at the corners (since there is no need and it simplifies the picture) but place the first anchor point at distance $\eps\tau/2$ from the corner. We also give alternating orientations to the anchor points, which determine the direction in which a tour can cross.  Although tours have no orientation in our definition of the VRP, we shall from now on assume that every tour comes with an orientation and crosses grid lines only at anchor points and in the specified direction. 
		
		\begin{definition}
			We say that a solution is \emph{anchor point respecting} if each tour crosses grid lines only at anchor points and any anchor point is crossed at most 6 times by each tour and only in the given direction. 
		\end{definition}
		
		\begin{figure}[!h]
			\centering
			\includegraphics[width=0.8\linewidth]{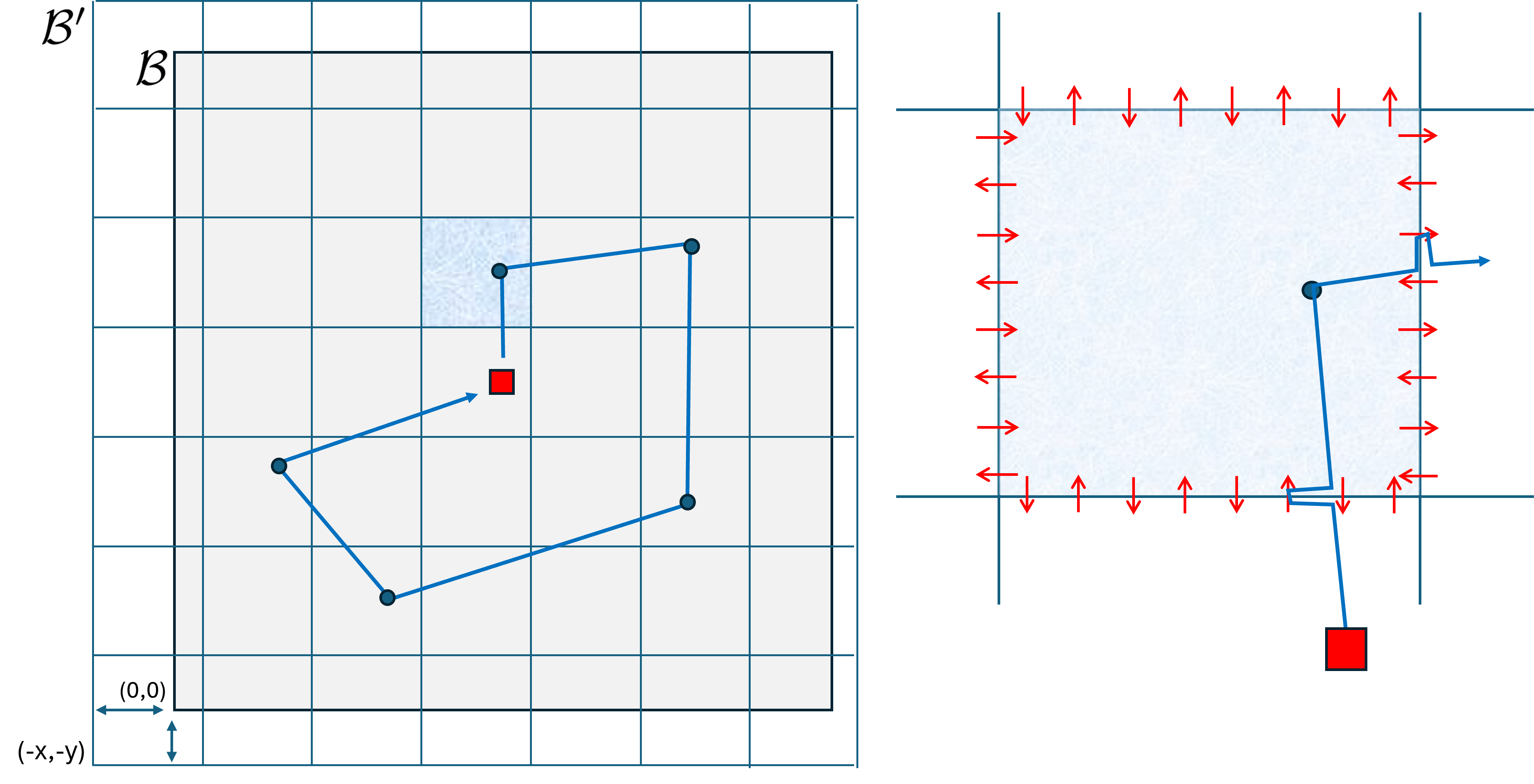}
			\caption{A grid of size $\eeps$ is placed at random position. Each grid square receives $1/\epsilon$ anchor points per side with alternating direction of crossing. Tours are only allowed to cross grid lines at anchor points in the specified direction. }
			\label{fig:squares}
		\end{figure}

		\begin{lemma}\label{lem:acnhorrespectingUB}
			Let $\OPT_A$ be the optimal cost over all anchor point respecting solutions. Then 
			\[\expected[\OPT_A]\le (1+5\eps) \OPT,\]
			where the expectation is over the random placement of the grid. 
		\end{lemma}
		\begin{proof}
			Consider an optimal solution $\newGamma$ and let $T$ be some tour in it. First, assume $T$ is not directed and ignore the direction of the anchor points. We make $T$ anchor point  respecting, as done in the PTAS for TSP~\cite{Arora98JACM}, by moving each crossing with a grid line to the nearest anchor point. The detour for each crossing is at most the inter anchor point distance: $\epsilon\tau$. As shown in~\cite{Arora98JACM}, it is possible to traverse $T$ such that each anchor point is crossed at most twice. So we use this traversal of $T$ and if an anchor point is crossed in the wrong direction then we move the crossing to an adjacent anchor point. The extra cost for the adjustment is at most $2\epsilon\tau$ per crossing. Now, each anchor point is crossed at most 6 times and only in the specified direction. 
			The total extra cost per crossing of $T$ with a grid line is at most $3\epsilon\tau$.
			 
			To compute the expected number of crossings of $\newGamma$ with grid lines observe that for any straight line segment between points $(x',y')$ and $(x'',y'')$, the expected number of crossings is exactly 
			\[
			|x''-x'|/\tau+|y''-y'|/\tau.
			\]
			Thus, the total expected detour due to making the solution  anchor point respecting is no more than \[(3\epsilon\tau)\sqrt{2}\OPT/\tau\le 5\epsilon \cdot \OPT.\]  
			\end{proof}\bigskip

		\begin{figure}[!h]
			\centering
			\includegraphics[width=0.5\linewidth]{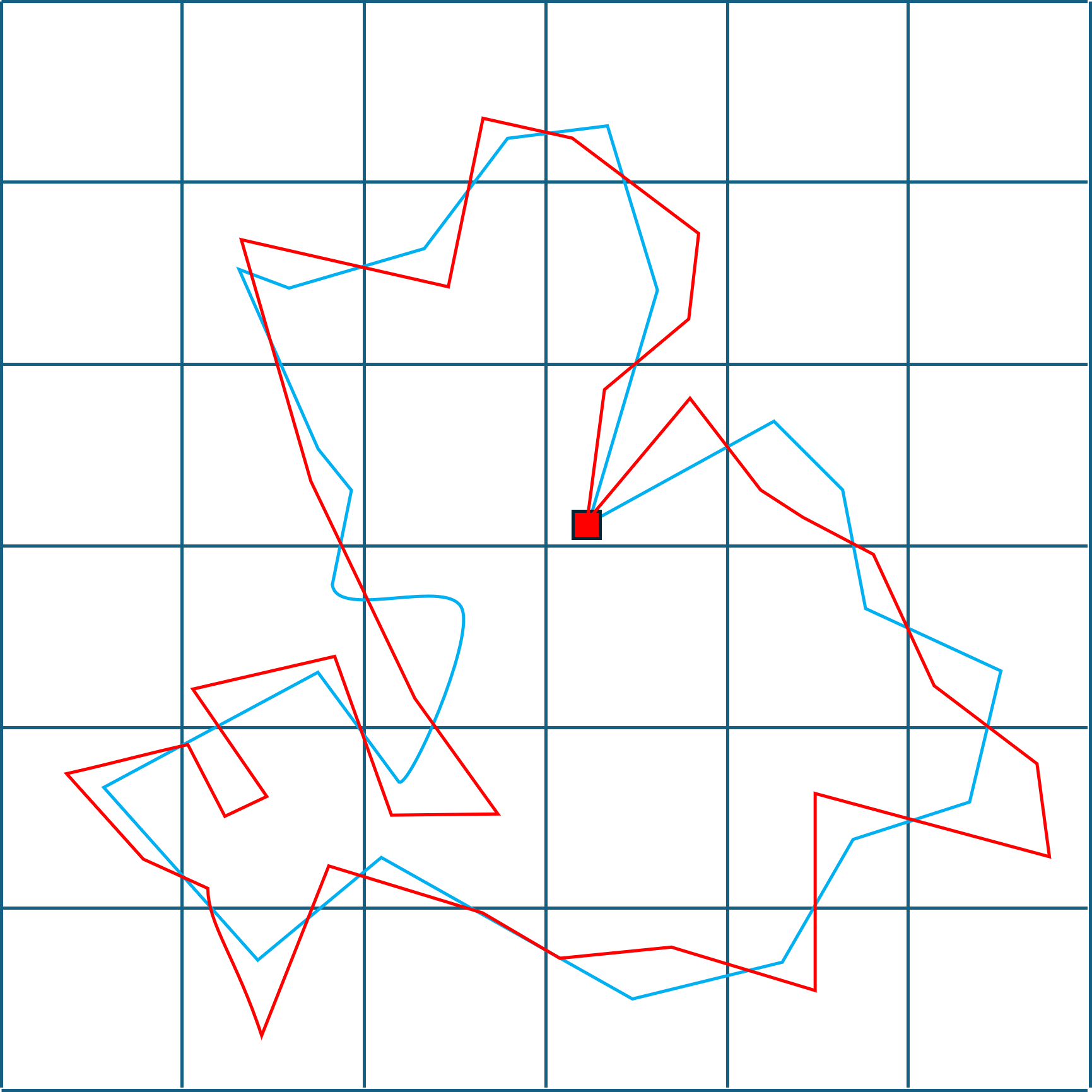}
			\caption{Tours of the same type have the same anchor point sequence.}
			\label{fig:same_type}
		\end{figure}
		
		From now on we take a fixed grid and we will construct a PTAS for finding the best anchor point restricting solution.
		For any tour, we can list the sequence of anchor points it crosses.  
		We say that two tours are of the same \emph{type} if their sequences of anchor points are the same.
		
		\begin{lemma}\label{lem:tourtypes}
			There are only $\eeps$ different tour types. 
		\end{lemma}
		\begin{proof}
			There are $\eeps$ squares and $\eeps$ anchor points per square and any tour crosses an anchor point $O(1)$ times. So, there are only $\eeps$ possible anchor point sequences.   
		\end{proof}\bigskip
		We shall enumerate all possible tour types $t$ and store this information in a list. More precisely, for any type $t$ and square $\sigma$ we keep a list of anchor point pairs 
		\begin{equation}\label{eq:tourtypes}
			(a^{(t,\sigma)}_i,b^{(t,\sigma)}_i)  \text{ for } i\in \{1,2,\dots,\npairs^{(t,\sigma)}\},
		\end{equation}
		for some $\npairs^{(t,\sigma)}=\eeps$, describing exactly how tours of type $t$ enter and leave $\sigma$. The list is empty if $\sigma$ is not crossed by tours of type $t$. 
		In general, we know that 
			$\npairs^{(t,\sigma)}\le 12/\epsilon$, since a square has $4/\epsilon$ anchor points and each anchor point respecting tour crosses each anchor point at most 6 times. This gives at most $6\cdot (4/\epsilon)/2=12/\epsilon$ pairs. In fact, we can strengthen this a bit. Looking at the proof of Lemma~\ref{lem:acnhorrespectingUB}, the undirected tour crosses each anchor point at most twice (in stead of 6 times) before replacing the crossings to neighboring anchor points. We conclude that 
			\begin{equation}\label{eq:npairs}
			\npairs^{(t,\sigma)}\le 4/\epsilon=\eeps.
		\end{equation}
		\begin{definition}\label{def:segment}
			A \emph{segment} for a square $\sigma$ and tour type $t$ is a set of 	$\npairs^{(t,\sigma)}$ paths, where there is one path between each anchor points pair  $a^{(t,\sigma)}_i,b^{(t,\sigma)}_i$ for all $i\in \{1,2,\dots,\npairs^{(t,\sigma)}\}$.
		\end{definition}	
			 
			 \subsubsection{Enumeration of subproblems}\label{sec:enumeration}
		For each tour type $t$ we guess more detailed information. By \emph{guess} we actually mean that we enumerate over all possible values.  
		For each tour type $t$ we guess:
		\begin{enumerate}[(1)]
			\item The number of tours, $m^{(t)}$, of type $t$.
			\item For each square $\ses$, the total number of points, say $n^{(t,\ses)}$, visited inside $\ses$ in total by tours of type $t$.
		\end{enumerate}
		By Lemma~\ref{lem:tourtypes}, the number of guesses is $n^{\eeps}$. We shall restrict to guesses that have the following obvious restrictions:
		
		\begin{enumerate}[(a)]
			\item $\sum\limits_t m^{(t)}\le n$. (No more than $n$ tours in total.)
			\item $n^{(t,\ses)}=0$ if $\ses$ is not crossed by tours of type $t$.
			\item $\sum\limits_\ses n^{(t,\ses)}_{\ses}\le c\!\cdot\! m^{(t)}$ for all $t$. (The average number of points over tours of type $t$ is at most $c$.)
			\item $\sum\limits_{t} n^{(t,\ses)}$ is equal to the number of points inside square $\ses$. 
		\end{enumerate}	
	We will enumerate over (1) and (2) with restrictions (a)-(d). For each choice of values, this gives one subproblem for each square $\sigma$, which we call the \emph{general $m$-paths problem}. Here, we need to find for every tour type a set of paths that together satisfy (1) and (2) with the restrictions (a)-(d). The cost is the total length of the paths.   
		 
		\bigskip\noindent{\textbf{General $m$-paths problem}}:\\[1mm]
		\noindent An instance is defined by:\vspace{-2.2mm}
		\begin{itemize}
			\setlength\itemsep{-0mm}
			\item[-] a square $\ses$,
			\item[-] the number of tours $m^{(t)}$ for each type $t$,
			\item[-] the number of points  $n^{(t,\sigma)}$ visited in $\ses$ for each type $t$, 
			\item[-] for each type $t$, a set of ordered pairs of anchor points  $(a^{(t,\sigma)}_i,b^{(t,\sigma)}_i)$  for $i\in \{1,2,\dots,\npairs^{(t,\sigma)}\}$, with  $\npairs^{(t,\sigma)}\le 4/\eps$.
		\end{itemize}
		A solution gives for each type $t$: \vspace{-2.2mm}
		\begin{itemize}
			\item[-] a set of $m^{(t)}$ paths from $a^{(t,\sigma)}_i$ to $b^{(t,\sigma)}_i$ for all $i\in \{1,2,\dots,\npairs^{(t,\sigma)}\}$, such that these $m^{(t)}\npairs^{(t,\sigma)}$ paths together visit exactly $n^{(t,\sigma)}$ points and any point in $\sigma$ is visited exactly once. 
		\end{itemize}
		The objective is to minimize the total length of the paths.  
		
		\bigskip 
		We show in Section~\ref{sec:mpaths} that the general $m$-paths problem admits a Q-PTAS.

		\subsection{Constructing tours from solutions to the subproblems.}
		Denote by $\Lambda$ the set of all enumerated combination of values 
		$m^{(t)}$ and $n^{(t,\ses)}$ for all $t$  as described in Section~\ref{sec:enumeration}. For each $\Lambda\in \Lambda$ and square $\sigma$, denote by $I(\lambda,\sigma)$ the corresponding instance of the  general $m$-paths problem. We shall refer to these $m$-paths problems as the \emph{subproblems}.  	By Lemma~\ref{lem:tourtypes}, the number of subproblems is $n^{\eeps}$.

Consider an optimal anchor point respecting solution for our bounded CVRP instance, with cost $\OPT_A$. Let $\lambda^*\in \Lambda $ be the corresponding information on $m^{(t)}$ and $n^{(t,\ses)}$ as described by (1) and (2). Let $\OPT_{\lambda^*}^{\sigma}$ be the optimal costs for the corresponding subproblems $I(\lambda^*,\sigma)$. Then   
		\[\sum_\sigma \OPT_{\lambda^*}^{\sigma}\le \OPT_A.\] 
		Note that the inequality may be strict since our subproblems are relaxations: for each square we only have a capacity constraint for each type and not for each individual path that crosses the square.

		Now assume that for some $\lambda\in \Lambda$ we found a solution to the subproblem $I(\lambda,\sigma)$ of cost $Z_{\lambda}^{(\sigma)}$ for all squares $\sigma$. In this section, we show how to construct a solution to the bounded CVRP instances  of cost at most $(1+\eps)\sum_{\sigma}Z_{\lambda}^{(\ses)}$. This implies that if we have an $\alpha$-approximation algorithm for the general $m$-paths problem, then if we guess $\lambda= \lambda^*$ we obtain a CVRP solution of cost at most 
		
		\[(1+\eps)\sum_{\sigma} Z_{\lambda^*}^{(\ses)} \le (1+\eps)\alpha  \sum_\sigma \OPT_{\lambda^*}^{\sigma}\le (1+\eps)\alpha\OPT_A.\]
		
For now, we fix some arbitrary $\lambda\in \Lambda$ and assume that we found a feasible solution for the corresponding general $m$-paths problem for each of the squares.
In the next section, we define a set of linear constraints which capture the type of bin packing problem that we need to solve. In Section~\ref{sec:almostcompletetours}, we will construct partial tours from the integral part of an extreme solution to the system of linear constraints. 
Then, in Section~\ref{sec:closingtours}, we first close the holes in the tours by connecting endpoints by straight line segments and then add the remaining points to the tours.
		
		\subsubsection{A linear programming formulation.}
		Assume that for some guess $\lambda\in \Lambda$ we found a feasible solution for the subproblem $I(\lambda,\sigma)$ on each square $\sigma$. The next step is to combine these solutions into tours. 
		This is a kind of bin packing problem. Note that for each type $t$, the total number of points visited is no more than $c\!\cdot\! m^{(t)}$, by restriction (c).  Further, by choosing the number of squares, $q^2$, in the grid large, each tour goes through many squares. So, intuitively, for $q$ large enough we can construct a feasible solution out of the solutions to the subproblems at small additional cost. We will do this independently for each type $t$. So, 
		we fix some type $t$ for now and drop the index $t$ from the notation. So let $m=m_t$ and denote $n^{\ses}=n^{(t,\ses)}$ for all $\ses$.

		Let $\mathcal{S}$ be the set of squares crossed by our fixed type $t$ and denote $|\SSS|=k$. Although $k$ depends on $t$, it is clear that $k\le q^2$, since this is the number of squares in the grid. On the other hand, since the instance is bounded and any point is at distance at least $\delta D$ from the origin, $k\ge \delta D/\tau=\delta(q-1)/2$. Hence, we have the following bounds on $k$ that we will only use at the very end of this section to obtain inequality~\eqref{eq:bound_on_m}. In essence, we can ensure a lower bound on $k$  by choosing $q$ large enough.
		\begin{equation}\label{eq:bound_on_k}
			\delta(q-1)/2\le k\le q^2, \text{ where $q=\eeps$ is chosen appropriately later}.
		\end{equation}
		\begin{figure}
			\centering
			\includegraphics[width=0.95\linewidth]{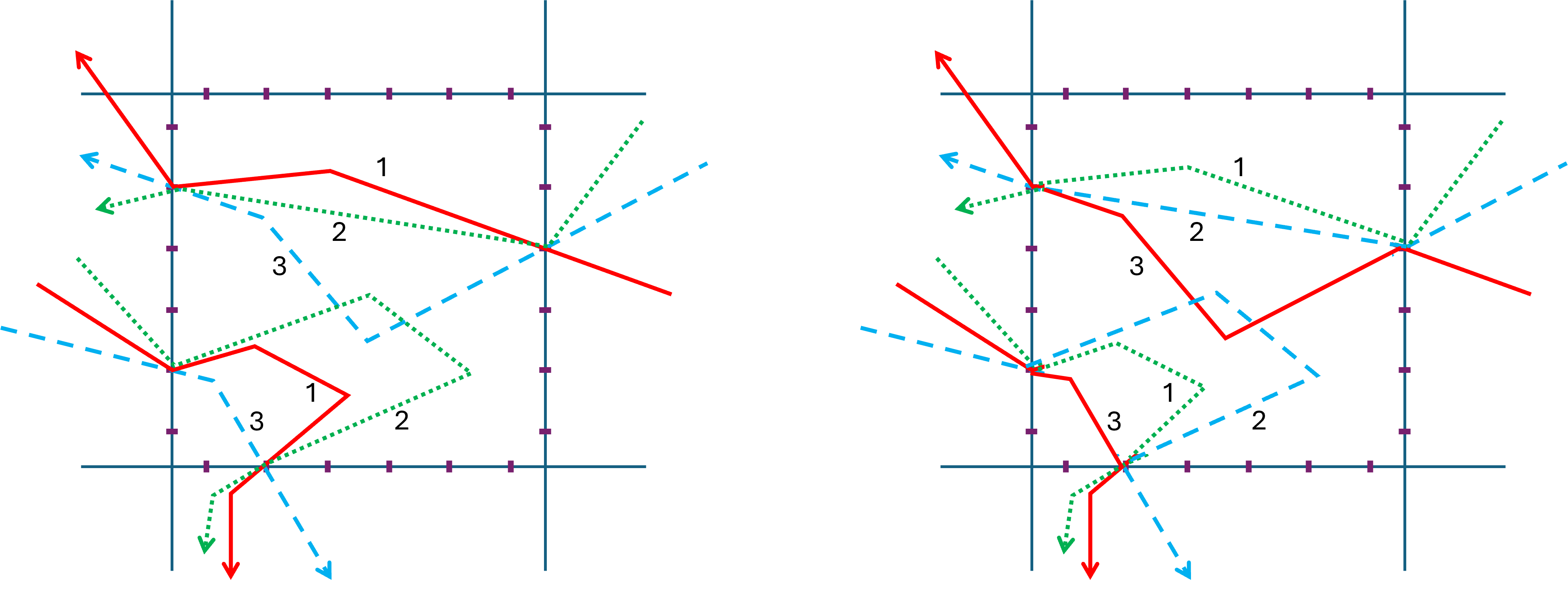}
			\caption{The square $\sigma$ is traversed twice by each of the $m=3$ tours of some type $t$. The 3 paths of each traversal are labeled 1,2,3.  Two paths with the same label define one segment. The capacity constraint of the CVRP is handled by swapping segments between tours. The figure shows two different matchings of segments to tours. }
			\label{fig:swapsegments}
		\end{figure}
			For now, consider some square $\ses\in \mathcal{S}$. The computed solution for $I(\lambda,\sigma)$ gives for each pair of anchor points a set of \emph{exactly} $m$ paths.  Arbitrarily label the $m$ paths between each anchor point pair by  $1,2,\dots ,m$ and let $n^{\ses}_{j}$ be the number of points visited by all paths with label $j$ inside $\ses$. 
		Then, the set of paths with label $j$ defines a segment inside $\ses$ and we refer to it as \emph{segment} $j$ in $\ses$. We denote the index set of the segments by $\PP=\{1,2,\dots,m\}$. For example, Figure~\ref{fig:swapsegments} shows 3 segments, labeled 1,2,3.

		We shall construct $m$ tours by matching segments to tours in each square. Let $\mathcal{T}=\{1,2,\dots,m\}$ be the index set of the tours of type $t$ to be constructed.  So when we write `tour $i$', then this  may either refer to a complete tour or to a tour with index $i$ that is to be constructed from the LP-solution. 
		(Figure~\ref{fig:swapsegments} shows two possible matchings of segments to tours.)
		
		Define variables $x(i,\ses,j)$ for all tours $i$, squares $\ses$, and segment labels $j$. Consider the following set of constraints:
		\begin{align}
			\sum_{j\in \PP} x(i,\ses,j)=1 && \text{ for all $\ses\in \SSS, i\in \TT$}\label{cs1} \\  
			\sum_{i\in \TT} x(i,\ses,j)=1 && \text{ for all $\ses\in \SSS, j\in \PP$} \label{cs2} \\  
			\sum_{\ses\in \SSS}\sum_{j\in \PP}   x(i,\ses,j)n^{\ses}_{j}\le c && \text{ for all $i\in \TT$} \label{cs3} \\  
			x(i,\ses,j)\ge 0 && \text{ for all $\ses\in \SSS$, $i\in \TT$, and $j\in \PP$}. \label{cs4} 
		\end{align}

		Constraint~\eqref{cs1} says that each tour $i$ selects exactly one segment from every square $\ses$.	Constraint~\eqref{cs2} says that each segment $j$ is selected exactly once in each square $\ses$ .
		Constraint~\eqref{cs3} says that each tour has at most $c$ points.

		At least one fractional solution exists, namely 
		$x(i,\ses,j)=1/m$ for all $i,\ses,j$. (Recall that $|\TT|=|\PP|=m$.)  
		The first two constraints are obviously true. The third constraint, \eqref{cs3}, follows from the obvious assumption we made on the guessed numbers, namely, that the average number of points visited by tours of the same type is no more than $c$.
		\[\sum_{\ses\in \SSS}\sum_{j\in \PP}   x(i,\ses,j)n^{\ses}_{j}=\frac{1}{m}\sum_{\ses\in \SSS}\sum_{j\in \PP} n^{\ses}_{j}\le \frac{1}{m} mc =c.\]

		An integral solution to this system of constraints has exactly $mk$ variables equal to 1. (Each tour is assigned one segment for every square.)  In general, extreme solutions are not integral and may have many fractional variables.
		To reduce the number of fractional variables in extreme solutions, we will replace constraint~\eqref{cs2}. The idea is to group, within each square, the segments containing at most $c$ input points such that any two segments within a group contain approximately the same number of points.
	    More precisely,  let $\beta=\eeps$, which we choose appropriately later in~\eqref{eq:frac_cost3}. Then, for each square $\ses\in\SSS$, we partition the set of 
	    $\{j\mid n^{\ses}_{j}\le c\}$ 
	    into groups $G_{\ses,1},G_{\ses,2},\dots$ such that within each group $G_{\ses,h}$ we have that 
	    \[|n^{\ses}_{j}-n^{\ses}_{j'}|\le \Delta:=\eps c/(\beta k),\] for any pair $j,j'\in G_{\ses,h}$.   
		Note that this will give no more than $c/\Delta=\beta k/\eps$ groups per square $\ses$.     
		We keep constraint~\eqref{cs2} for segments with more than $c$ points: 
		\begin{align}
		\sum_{i\in \TT} x(i,\ses,j)=1 && \text{ for all $\ses\in\SSS, j\in \PP$ with $n^{\ses}_{j}> c$ } \label{cs2a} 
		\end{align}
		For the other segments, we replace constraint~\eqref{cs2} by grouped constraints. Let $g_{\ses}\le \beta k/\eps$ be the number of groups in $\ses$.    
		\begin{align}
		\sum_{i\in \TT} \sum_{j\in G_{\ses,h}} x(i,\ses,j)= |G_{\ses,h}|&&\text{ for all $\ses\in\SSS$ and $h\in \{1,2,\dots,g_{\ses}\}$}\label{cs2b} 
		\end{align}
		
		The system of linear constraints has at least one solution since taking all $x$-values $1/m$ is still feasible. 
		Let's count the number of constraints in the system defined by~\eqref{cs1},\eqref{cs3},\eqref{cs4},\eqref{cs2a}, and~\eqref{cs2b}. There are $mk$ constraints of type~\eqref{cs1} and the number of constraints~\eqref{cs3} is $m$. 
		The number of constraints~\eqref{cs2b} is at most $\beta k^2/\eps$.
		The number of constraints~\eqref{cs2a} is at most $m$, since there are at most $m$ segments with more than $c$ points by our restriction on the average number of points per tour. See property (c) in Section~\ref{sec:enumeration}.  
		Hence, any extreme solution has at most 
		\begin{equation}\label{eq:n_support_of_LP}
			mk+ 2m+\beta k^2/\eps
		\end{equation} strictly positive values. 
		In Section~\ref{sec:almostcompletetours}, we will construct a \emph{partial solution} from the integral part of an extreme solution $x$.

		\subsubsection{Constructing almost complete tours from an LP-solution}\label{sec:almostcompletetours}
		Let $x$ be an extreme solution for the system given by ~\eqref{cs1},\eqref{cs3},\eqref{cs4},\eqref{cs2a}, and~\eqref{cs2b}. 
		Say that a pair $(i,\ses)$ is \emph{integer} if there is a $j$ such that  $x(i,\ses,j)=1$ and (by Constraint~\eqref{cs1})  $x(i,\ses,j')=0$ for all $j'\neq j$. Say that a pair $(i,\ses)$ is \emph{fractional} otherwise. A fractional pair $(i,\ses)$ gives at least two fractional variables  $x(i,\ses,j)>0$. 
		Hence, by~\ref{eq:n_support_of_LP}, at least $mk- (2m+\beta k^2/\eps)$ variables are equal to 1 and at most $2m+\beta k^2/\eps$ pairs $(i,\ses)$ are fractional.
		Let $\textsc{Int}$ be the set of integer pairs $(i,\ses)$ and $\textsc{Frac}$ be the set of fractional pairs. As argued above:
		\[|\textsc{Int}|\ge  mk- (2m+\beta k^2/\eps).\]
		\[|\textsc{Frac}|=mk-|\textsc{Int}|\le 2m+\beta k^2/\eps.\]
			N.B. At this point it is good to realize that $k=\eeps$ which we can control by choosing the size $q=\eeps$ of the grid. (See~\eqref{eq:bound_on_k}.) So, for large enough $k$ and $m\gg k$  we have $|\textsc{Frac}|/|\textsc{Int}|\rightarrow 0$. 
			We will work out the details below.

		Consider a square $\ses$ and all tours $i$ such that  $(i,\ses)$ is an integer pair as defined above. 
	Due to the grouping, we no longer have a matching of segments to tours and it may happen that different tours $i$ are assigned, by the LP, to the same segment $j$. But by ~\eqref{cs2b}, we can easily fix this by reassigning within each group. We then get a solution, say $x'$, that keeps all pairs $i,\ses$ integer (but may assign $(i,\ses)$ to another $j$ within a group) and such that Constraints~\eqref{cs1} and~\eqref{cs2} hold and capacity constraint~\eqref{cs3} is violated by at most 
		\[\sum_{\ses\in\SSS} \Delta=k(\eps c/(\beta k)=(\eps/\beta)c.\] 
			We say that the \emph{integral} part of tour $i$ is composed of the segments  $(\ses,j)$ for which $x'(i,\ses,j)=1$.
		At this point we have $m$ \emph{partial tours}, composed of the segments defined by integral part of $x'$, and which violate the capacity constraint by at most a factor  
		$1+\epsilon/\beta $. 
		Next, we remove some segments from each partial tour that has more than $c$ points. Then, in Section~\ref{sec:closingtours} we add all the remaining points to the tours.  
		
		A partial tour that is composed of $u$ segments has at most $c$ points after removing the $\lceil(\epsilon/\beta)u\rceil\le (\epsilon/\beta)u +1$ segments with the largest number of points.	
		Taking the sum over all $m$ tours we see that we need to remove at most   
		$(\epsilon/\beta)|\textsc{Int}|+m$ segments from all partial tours.
		Let $\textsc{Int}'$ be the set of integral pairs $(i,\ses)$ that remain after deletion. 
		\begin{eqnarray}\label{eq:INT'}
			|\textsc{Int}'|&\ge & |\textsc{Int}|-((\epsilon/\beta)|\textsc{Int}|+ m)\nonumber \\
			&= &(1-\epsilon/\beta) |\textsc{Int}|-m \nonumber \\
			& \ge & (1-\epsilon/\beta)\left(mk - (2m+\beta k^2/\eps)\right)-m\nonumber\\
			& = & (1-\epsilon/\beta)mk - (1-\epsilon/\beta) (2m+\beta k^2/\eps )-m\nonumber\\
			& \ge  & (1-\epsilon/\beta)mk  - \beta k^2/\eps-3m.
		\end{eqnarray}
		At this point we constructed $m$ partial tours of at least a total number of segments given by~\eqref{eq:INT'} which all satisfy the capacity constraint. In the next section we close the tours and divide the remaining segments fractionally over the tours.

		\subsubsection{Closing the tours and adding the remaining segments}\label{sec:closingtours}
			\begin{figure}
			\centering
			\includegraphics[width=0.7\linewidth]{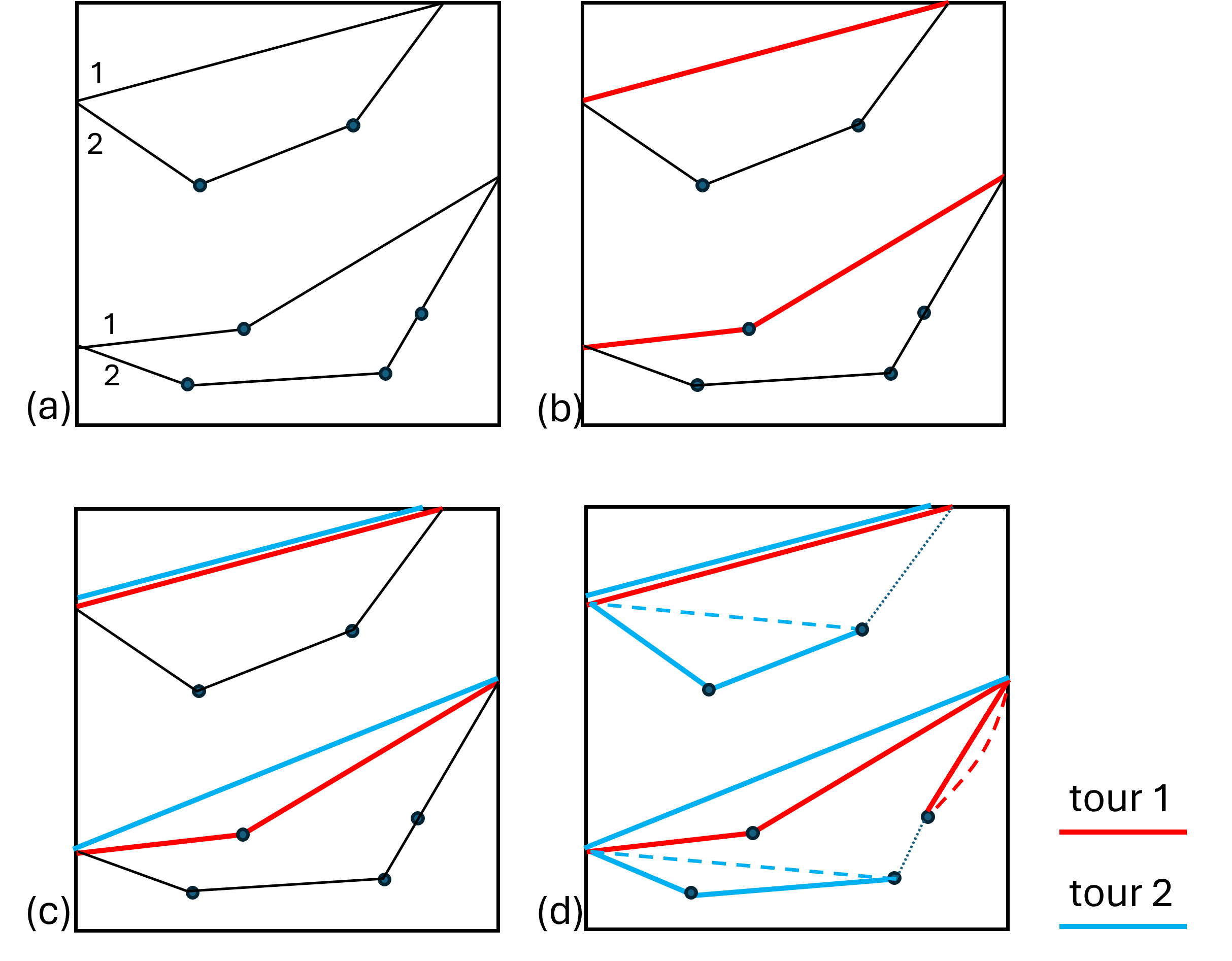}
			\caption{Figure (a) shows $m=2$ segments, labeled 1 and 2. Figure (b) shows the situation after the partial tour construction, where Segment 1 is assigned to the red tour (tour 1) and Segment 2 is not assigned. Figure (c) shows the situation after closing the tours. Tour 2 (blue)  is assigned a straight line segments between each of the two anchor point pairs. The 5  points on Segment 2 are not visited yet. Figure (d) shows how these remaining points  are included into tours 1 and 2 at small extra cost. }
			\label{fig:tourconstruction}
		\end{figure}
		
		Our partial solution is composed of $|\textsc{Int}'|$ segments. The   number of unassigned segments is 
		\begin{equation}\label{eq:bound_on_f}
			f:=mk-|\textsc{Int}'|\le (\epsilon/\beta)mk  + \beta k^2/\eps+3m.
		\end{equation}
		First, we close all tours at minimum cost, not using any of the points of the unassigned segments (Figure~\ref{fig:tourconstruction}(c)), and then we add the remaining segments to the tours (Figure~\ref{fig:tourconstruction}(d)). 
		Remember that any segment is composed of at most $4/\epsilon$ paths between pairs of anchor points. 
		(See the discussion at the beginning of Section~\ref{sec:enumeration}.)
		Hence, closing the gaps in the tour costs no more than $(4/\epsilon)\sqrt{2}\tau$ per missing segment, which we roughly bound by  $6\tau/\epsilon$.
		Closing all gaps has total cost at most 
		\begin{equation}\label{eq:frac_cost1}
			(6\tau/\epsilon)f.
		\end{equation}
		At this point, our $m$ partial tours are proper tours of type $t$, satisfying the capacity constraint. It remains to assign the yet unassigned segments to these tours. There are $m$ tours and $f$ unassigned segments. We simply assign them one by one fractionally to tours until the capacity $c$ of the tour is reached. Let us label the unassigned segments for now by $1,2,\dots,f$, and let $z_{ij}$ be the number of points from segment $j$ assigned to tour $i$. Then, this greedy assignment gives at most $m+f-1$ strictly positive values $z_{ij}$. (Like greedily and fractionally packing $f$ items into $m$ bins.) Call these the \emph{slices}. So a slice is part of a segment (or the complete segment) assigned to a tour. 
		Since any segment consists of at most  $4/\epsilon$ paths, a slice consists of at most  $4/\epsilon$ pieces of a path. Since all tours are closed and go through all squares in $\SSS$, we can add a slice to a tour at cost no more than $(4/\epsilon)2\sqrt{2}\tau\le (12\tau/\epsilon)$. (See the dotted lines in Figure~\ref{fig:tourconstruction}(d).) 
		Note however that the cost for the slices themselves is already counted, since these are part of the segments.   
		So, the total extra cost for inserting the $f$ remaining segments is at most
		\begin{equation}\label{eq:frac_cost2}
			(12\tau/\epsilon) (m+f-1).
		\end{equation}
		We now choose $\beta=72/\epsilon$. (The constant 72 is a bit arbitrary but gives a round number in the computation below.) 
		Using the bound~\eqref{eq:bound_on_f} on $f$, the total extra cost, as given by~\eqref{eq:frac_cost1} and~\eqref{eq:frac_cost2}, is bounded by
		\begin{eqnarray}\label{eq:frac_cost3}
			(18\tau/\epsilon) (m+f)&\le &\frac{18\tau}{\epsilon} \left((\epsilon/\beta)mk  + \beta k^2/\eps+4m \right)\nonumber \\
			&=  & 18\tau( mk/\beta  + \beta k^2/\eps^2+4m/\epsilon  )\nonumber \\
			&=  & \left(\epsilon{mk}/{4} +  1296 k^2/\eps^3+4m/\epsilon  \right)\tau
		\end{eqnarray}
		\begin{figure}[!h]
			\centering
			\includegraphics[width=0.45\linewidth]{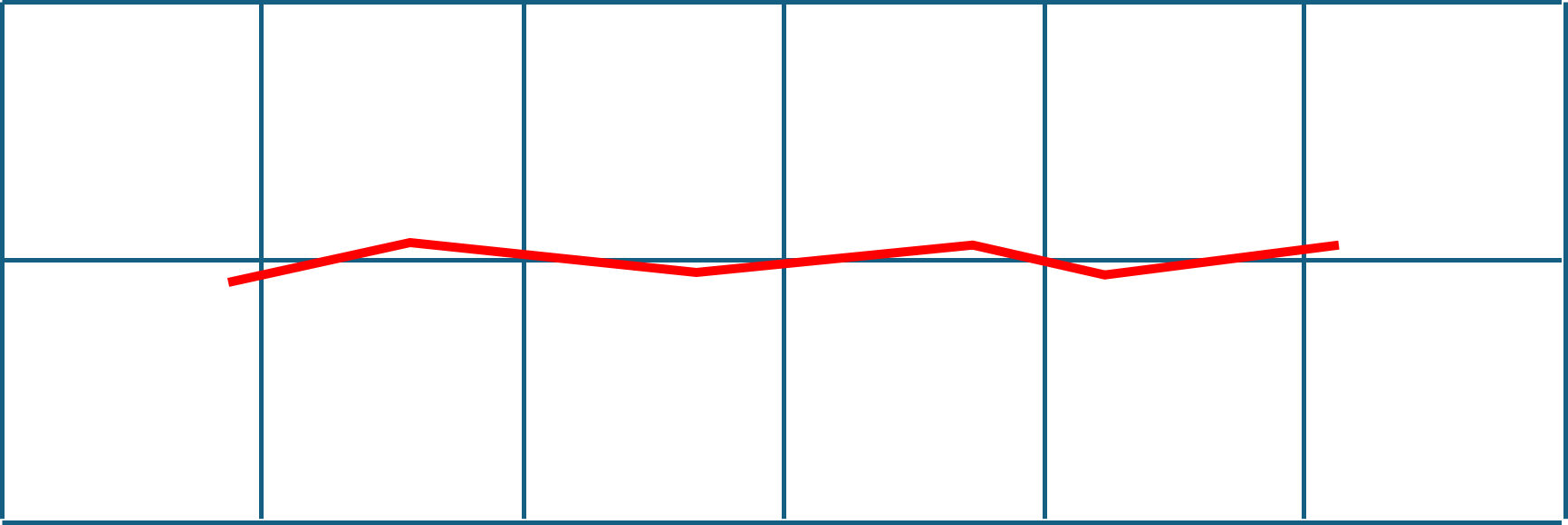}
			\caption{A path that goes through $k$ squares (of a grid) has length at least  $(k/2-2)$ times the side length of a square.}
			\label{fig:LBopt}
		\end{figure}
		
		Note that the cost for each tour of this fixed type $t$ we consider is at least $(k/2-2)\tau$ (See Figure~\ref{fig:LBopt}.) 
		Hence, the total cost for all tours of this fixed type $t$ is at least 
		\begin{equation}\label{eq:bound_on_OPTt}
			LB_t:= m(k/2-2)\tau.
		\end{equation}
		
		Comparing~\eqref{eq:frac_cost3} and~\eqref{eq:bound_on_OPTt} we see that the extra cost is at most $\epsilon LB_t$ if 
	\begin{eqnarray*}
		\left(\epsilon{mk}/{4} +  1296 k^2/\eps^3+4m/\epsilon  \right)\tau &\le& \epsilon  m(k/2-2)\tau\\
		\Leftrightarrow  1296 k^2/\eps^3+4m/\epsilon   &\le& \epsilon  m(k/4-2)\\
			\Leftrightarrow 1296 k^2/\eps^3 &\le& \epsilon  m(k/4-4/\epsilon^2)
	\end{eqnarray*}
		If $k> 	16/\eps^2$ then $k/4-4/\epsilon^2>0$ and the above  is equivalent to  
		\begin{equation}\label{eq:bound_on_m}
			m\ge \frac{ 1296 k^2/\eps^3}{\epsilon(k/4-4/\epsilon^2)}.
		\end{equation}
	Now, we will ensure that $k> 16/\eps^2$ by using~\eqref{eq:bound_on_k} and choosing $\delta(q-1)/2>16/\eps^2$, i.e., we choose the grid size parameter
			\begin{equation*}\label{eq:bound_on_}
				q>32/(\eps^2\delta)+1.
			\end{equation*}	
		Remember, $m$ depends on the type $t$ and we guess a value $m_t$ for each type $t$. 
		Also, we know for each type the number $k$ of squares that a tour of that type goes through. 
		Hence, if we guess $m_t$ to be at least the bound~\eqref{eq:bound_on_m} above then the extra cost for constructing tours out of the LP solutions is at most $\epsilon$ times the lower bound~\eqref{eq:bound_on_OPTt}. 
		If, on the other hand, our guess for  $m_t$ is less than the bound~\eqref{eq:bound_on_m}, then we say that type $t$ has a \emph{small} number of tours and then our analysis in the cost for constructing feasible tours fails in this case. This is not a problem though since then $m_t=\eeps$ and we can afford to guess a lot more information for these tour types. We guess for each tour of type $t$  and each square $\ses$,  exactly the number of points visited by the tour in the square. The number of guesses  is $n^{\eeps}$. 
		We shall add this information as a hard constraint  to  subproblems $I(\lambda,\sigma)$. Hence, for types $t$ with a small number of tours, there is no correction needed as these will always have at most $c$ points per tour. 
		
		Summarizing, we enumerated over all possible combinations of values 
		$m^{(t)}$ and $n^{(t,\ses)}$ as described in Section~\ref{sec:enumeration}. The set of combinations is denote by $\Lambda$ and has size $n^{\eeps}$.
		Each $\lambda\in \Lambda$ defines a subproblem $I(\lambda,\sigma)$ for each square $\sigma$. We solve all of these. For every $\lambda\in\Lambda$ for which we find a solution for every $\sigma$, we solve a bin packing LP. The integral part of the LP-solution defines a partial tour and we add the remaining points to tours in a greedy way. The running time of the whole reduction from bounded instances to the $m$-paths problem is $n^{\eeps}$. 
			
		\begin{theorem}
			There is a $(1+\epsilon)\alpha_{\epsilon}$ approximation scheme for CVRP with running time $n^{\eeps}f(n,\epsilon)$, where $f(n,\epsilon)$ is the time of an $\alpha_{\epsilon}$-approximation algorithm for the general $m$-paths problem on an instance of $n$ points. 
		\end{theorem}

		\section{An $n^{\tilde{O}(\log n)}$-time approximation scheme}\label{sec:mpaths}
		
		In the previous section, we defined the general $m$-paths problem and reduced our search for an approximation scheme for the Euclidean CVRP to an approximation scheme for this problem. We will first develop an approximation scheme for the special case where the square $\sigma$ is crossed by only one type of tour and this happens only once. \\

		\noindent\textbf{$m$-paths problem}\\
		Given is a set of $n$ points\footnote{Note that the notation $n$ for the number of points is ambiguous since we use it as well for the number of points in the complete CVRP instance. In this section, we only focus on the $m$-paths problem so there should be no confusion.} 
		in the Euclidean plane and two dedicated points $a,b$, plus a number $m\le n$. A solution  is a set of exactly $m$ paths between $a$ and $b$ that cover all the given points. Goal is to minimize the total length of the paths. \\

	The case $m=1$ corresponds to the TSP-path problem and the case $a=b$ corresponds to the TSP problem for any $m$, since an optimal solution will use only one tour. The problem appears to be substantially more difficult for arbitrary $a,b$, and $m$ and we are only able to obtain a Q-PTAS here. In the context of the CVRP reduction we may assume that:
	\begin{itemize}
		\item all points are in a square $\sigma$ of side length $\tau$, and
		\item  $a,b$ are on the boundary of $\sigma$ at a distance of at least $\epsilon \tau/2$.  
	\end{itemize}
 To see the second property, note that we gave an orientation to the anchor points and the distance between $a$ and $b$ is at least the minimum distance between any two anchor points. This minimum is attained for anchor points near the corner of the square and is  $\sqrt{2}\eps\tau/2>\eps\tau/2$. This minimum distance immediately gives us a useful lower bound on the optimal cost.  
	 	Let $\OPT_{\sigma}$ be the optimal cost to an instance of the $m$-paths problem on square $\sigma$.
	Then,      
		\begin{equation}\label{eq:LB_OPTpaths}
			\OPT_{\sigma}\ge m\eps\tau/2. 
		\end{equation}
		The reader may want to verify that for $m=\eeps$, the PTAS by Arora~\cite{Arora98JACM} for Euclidean TSP immediately gives a PTAS for the $m$-paths problem. In stead of one tour we need to construct $m$ paths which can be done using the same approach at the expense of a factor $n^{\eeps}$ in running time.  
		However, $m$ is not bounded in our case and we shall assume $m\ge 18/\epsilon$ from now on. Also, we assume from now that $m\le n$, since if $m>n$ there are at least $m-n$ paths with no points and we can directly reduce to the case $n=m$. 
		
				A relatively simple adjustment of the algorithm by Arora does lead to a Q-PTAS with running time $n^{\tilde{O}(\log n)}$, which is what we present in the next section. Note that this is already a huge improvement over existing algorithms~\cite{DasMathieu2014},\cite{JayaprakashS2023}. 
		It supports the intuition that the $m$-paths problem is significantly easier to solve than the CVRP and shows the value of the reduction from CVRP to $m$-paths. The Q-PTAS of this section is the basis for the more complex $n^{\tilde{O}(\log\log n)}$ version of Section~\ref{sec:loglog}.  
	
		\subsection{An $n^{\tilde{O}(\log n)}$ time algorithm for the $m$-paths problem}\label{sec:logn}
		
	Before going over the Q-PTAS in detail, we give a (very) brief explanation for the reader who is familiar with the PTAS from~\cite{Arora98JACM}.
		
		\paragraph*{The Q-PTAS in a nutshell} In stead of one TSP-tour, we have $m$ paths visiting $n$ points. 
		We assume the $m$ paths are directed from $a$ to $b$. Since all paths have the same endpoints and there is no capacity constraint on individual paths, there is no need for paths to cross each other.     
		Inside a dissection square $S$, the paths that cross the dissection square define a directed outer planar graph on the portals of $S$. We call this a \emph{flow graph} for the dissection square. There are $\tilde{O}(\log n)$ portals per dissection square, which gives at most  $\tilde{O}(\log n)$ arcs in any flow graph and $n^{\tilde{O}(1)}$ possible flow graphs.  
		It is easy to show that the flow through each portal is at most $m+n\le 2n$, which gives at most $n^{\tilde{O}(\log n)}$ ways to label the $\tilde{O}(\log n)$ arcs of the flow graph with a flow value. Hence, the size of the DP table is $n^{\tilde{O}(\log n)}$. We explain it in detail below. 
		
		\subsubsection{Grid, dissection squares, and portals}

		\paragraph*{Random grid and moving to mid points}
		Denote by $B_1$ the square $\sigma$ of side length $\tau$ of the instance. We rescale distances and define the side length of the square as $L_{1}=2^{\rho-1}$, where 
		\begin{equation}\label{eq:rho}
			\rho:=\lfloor\log_2 (n/\epsilon)\rfloor.
		\end{equation}
		Note that this implies $n/4\le \eps L_1\le n/2$. We put a square $B_0$ of side length $L_0=2L_1=2^{\rho}$ on top of $B_1$ at random position. More precisely, if $(0,0)$ is the lower left corner point of $B_1$ then $(-x,-y)$ is the lower left corner point of $B_0$ where $x$ and $y$ are taken uniformly at random from $\{1,2,\dots,L_1\}$.
		We define a unit grid on top of $B_0$ and move all points to the middle of the $1\times 1$ squares, which we shall call \emph{grid cells} from now. The error due to this rounding is at most $\sqrt{2}n\le 4\sqrt{2}\eps L_1\le 6\eps L_1$.
		Next, we use the assumption $m\ge 18/\epsilon$ and conclude from~\eqref{eq:LB_OPTpaths} (with $L_1\cong\tau$) that the error due to rounding to centers grid cells is at most
		\[6\eps L_1\le 18\OPT_{\sigma}/m\le \eps\OPT_{\sigma}.\] 
		
		\begin{figure}[!h]
			\centering
			\includegraphics[width=0.75\linewidth]{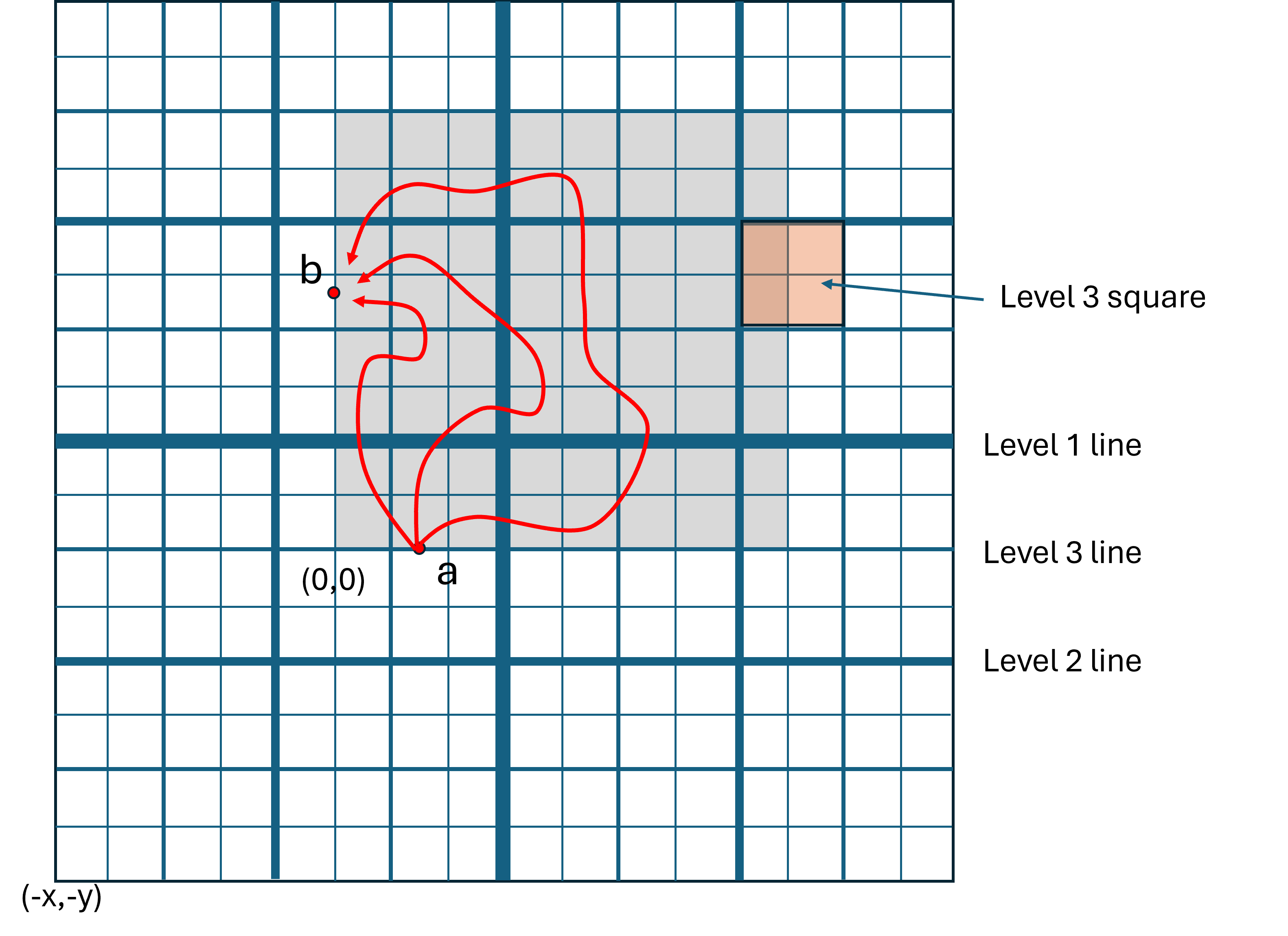}
			\caption{The square $B_0$ is placed at random position on top of  square $B_1$, containing the input points. }
			\label{fig:dissection-a}
		\end{figure}

		\paragraph*{Dissection and portals}
		We use exactly the same structure as in \cite{Arora98JACM}\footnote{We shall not restrict solutions to so called 'light' paths (which cross the side of a square only at $\eeps$ portals~\cite{Arora98JACM}) but we use the simpler version with $\tilde{O}(\log n)$ portal crossings. There is no advantage for using 'light' paths in our analysis: The union of a set of light paths is in general not light, since each path uses it own set of $\eeps$ portals.} for the dissection squares and portals. (See Figure~\ref{fig:dissection-a}.) The square $B_0$ is said to be of level 0 and the horizontal line and vertical line in the middle of $B_0$ are called level 1 lines. These lines split $B_0$ in four level $1$ squares. The  lines that fall in the middle of a level 1 square are of level 2 and together they split each level 2 square in four level 3 squares and so on. The grid cells are of level $\rho$. Also, the lowest level of a grid line is $\rho$. Note that a level $i$ square is bounded by two level $i$ lines and two lines of a higher level. (Where we consider the boundary of $B_0$ as level 0 lines.)       
		There are $O(\log n /\eps)$ portals per side of a dissection square exactly as in ~\cite{Arora98JACM}.
		To be precise, on a level $i$ line we place $1+2^{i+2}\rho/\eps$ equally separated points (the portals). 
		So, the inter portal distance of a level $i$ line is 
		\begin{equation}\label{eq:eta}
		\eta_i:=2^{\rho}/(2^{i+2}\rho/\eps)=\eps2^{\rho-i-2}/\rho.
		\end{equation}
		Portals in $B_0\setminus B_1$ are not really needed but we just keep them to make the description of the DP consistent. 
		Let $L_i$ be the side length of a level $i$ dissection square. Then 
		\begin{equation}\label{eq:Li}
		L_i=L_0/2^i=2^{\rho-i}.
		\end{equation}
		A level $i$ line is partitioned into $2^i$ line-segments, where each line-segment is the side of two adjacent level $i$ dissection square. So the number of portals per side of a level~$i$ square that is part of a level $i$ line is 
		\begin{equation}\label{eq:n_port}
			\nport:=L_i/\eta_i +1=4\rho/\eps+1 \ \text{, which is independent of $i$.}
		\end{equation}
		Hence, the number of portals on the boundary of any level $i$ square is $O(\rho/\eps)=\tilde{O}(\log n)$. 
		
		Portals on the intersection of grid lines are defined as in Figure~\ref{fig:cornerportal}. That means, a portal is only used for a crossing between neighboring squares, not for a diagonal crossing of the intersection point. So, we actually have 4 portals on each crossing.
		\begin{figure}[!h]
			\centering
			\includegraphics[width=0.4\linewidth]{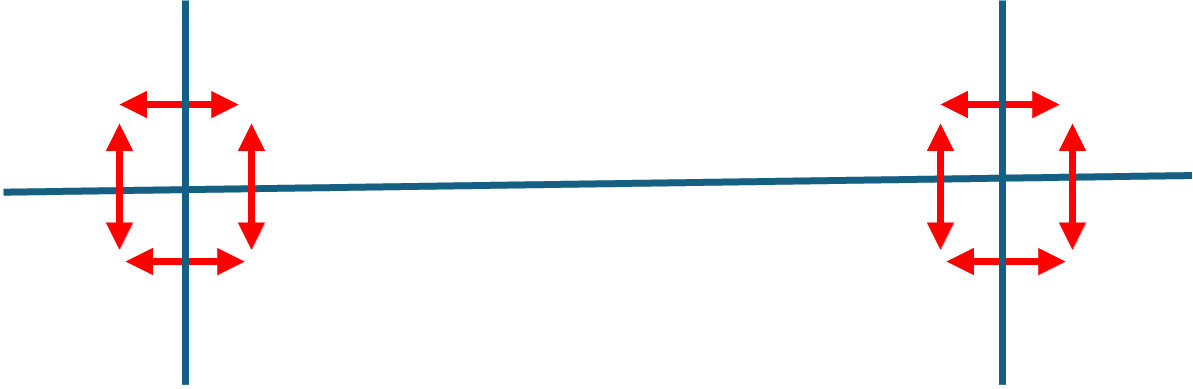}
			\caption{There are 4 portals on each intersection of two grid lines. The figure shows 2 intersections with each 4 portals.}
			\label{fig:cornerportal}
		\end{figure}

		\paragraph*{Portal respecting solutions}
		We will show that an optimal solution can be made portal  respecting at an expected increase in cost of at most a factor $1+\eps$. Again, this is basically the same analysis as~\cite{Arora98JACM}. The only difference is that we have $m$ paths in stead of one tour. The distance between the endpoints $a,b$ is much larger than the grid line distance though, so the effect of having paths in stead of tours is small. For completeness, we give the computation and our minor adjustments here.
		
		First we make the anchor points $a,b$ portal respecting. Anchor points do, in general, not coincide with a portals but we may assume they do by just moving both  anchor points to the nearest portal on the boundary of $B_1$. To see that this has no significant effect in cost, note that the inter \emph{anchor point} distance is $\eps L_1$ while the maximum inter \emph{portal} distance (i.e., of a level $1$ line) is  $\eps L_1/(4\rho)$, which is a factor $\log (n/\eps)$ smaller. 
		
		For the crossings of paths with grid lines we argue as follows. 
		The probability that a grid line $l$ is of level $i$ is  
		exactly $\frac{2^{i-2}}{2^{\rho-1}}$ for $i\ge 2$ and is $\frac{2^{0}}{2^{\rho-1}}$ for $i=1$. So, in general this probability is at most $\frac{2^{i-1}}{2^{\rho-1}}=\frac{2^{i}}{2^{\rho}}$. Hence, for any grid line $l$, the expected distance between the portals is bounded as follows. 
		\begin{eqnarray*}
			\text{Expected(Inter portal distance)} &=&\sum_{i=1}^{\rho} \text{Probability($l$ is of level $i$)}\cdot \eta_i\\
			&\le &\sum_{i=1}^{\rho}  \frac{2^{i}}{2^{\rho}}\eta_i=\sum_{i=1}^{\rho}  \frac{2^{i}}{2^{\rho}}\frac{\eps 2^{\rho-i-2}}{\rho} =\eps/4.
		\end{eqnarray*}
		As a result, the expected detour for moving the crossing of a path with a grid line to the nearest portal on that line is also at most $\eps/4$. 
		
		As in~\cite{Arora98JACM}, the number of crossings with grid lines is in the order of the optimal cost. 
		To see this, consider any path from $a$ to $b$ in an optimal solution.
		The path is a set of straight connections of length at least 1, except possibly for the two parts from $a$ to the first point on the path and from the last point on the path to $b$. Note however, that the distance between $a$ and $b$ is at least $\eps L_1/2=\epsilon\Omega(n)$. (See~\eqref{eq:rho}.) Therefore, we can safely say that 
		any path of length $z$ in an optimal solution has at most $4z$ crossings with grid lines. (The factor 4 is a rough upper bound.) So the number of crossings with grid lines in any optimal solution is no more than $4\OPT_{\sigma}$.
		
		From now on we denote by $\OPT$, the optimal cost among the portal respecting solutions for the $m$-paths problem with anchor points $a$ and $b$. The expected cost  over the random placement of the box $B_0$ is
		\begin{equation}\label{eq:exp_opt1}
			\text{Expected}(\OPT)\le \OPT_{\sigma}+(\eps/4)4\OPT_{\sigma}=(1+\eps)\OPT_{\sigma}.
		\end{equation}
		We will show that an optimal portal respecting solution for the $m$-paths problem can be found in quasi polynomial time. 
		
		\paragraph*{Shifting the anchor points}	
		To keep the description of the dynamic program consistent, we move the start- and  endpoint of the paths to the boundary of $B_0$. We fix a path from anchor point $a$ to some point $a'$ on the boundary of $B_0$ and fix a path from $b$ to some point $b'$ on the boundary of $B_0$. (See Figure~\ref{fig:dissection}.) We refer to these paths as \emph{anchor paths}. Each path is copied $m$ times. We may directly assume that these paths cross grid lines at portals since the cost of these paths will be fixed in the DP. Their only purpose is to move $a$ and $b$ to the boundary, which simplifies the description of the DP.\footnote{In stead of using these fixed anchor paths, one might also consider another approach, where  flow values of dissection squares in the  DP are based on containment of $a$ and $b$. For example, a dissection square containing $a$ but not $b$ should have an outflow of $m$.}     Let $Z^{0}$ be the total cost of this fixed part of the solution, i.e, $Z_0$ equals $m$ times the length of these two paths. 
		
			\begin{figure}[!h]
			\centering
			\includegraphics[width=0.45\linewidth]{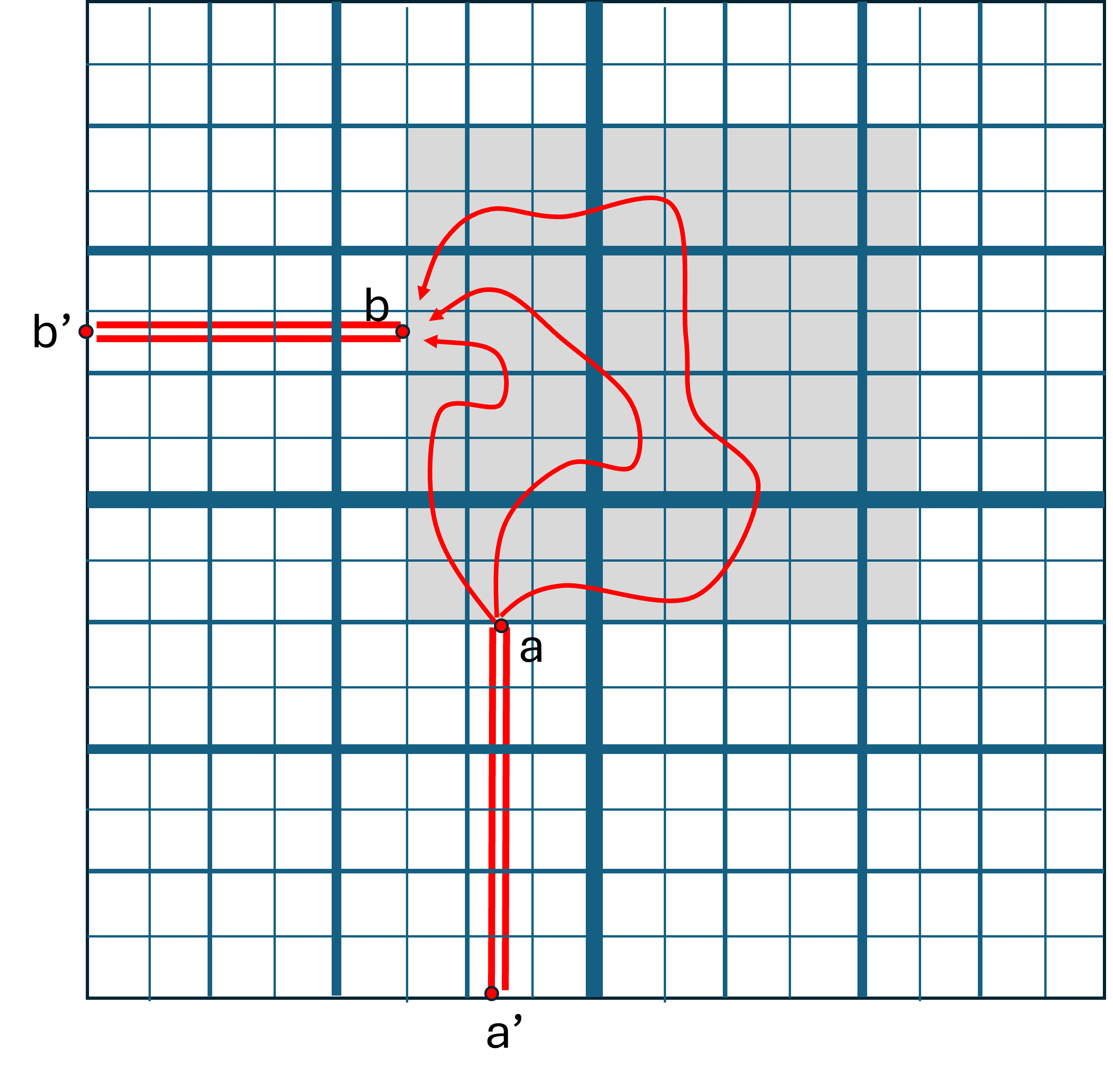}
			\caption{Source $a$ and sink $b$ are connected to the boundary by a path that we take as fixed part of the solution in the dynamic program. The total cost of these fixed paths is defined as $Z^0$.}
			\label{fig:dissection}
		\end{figure}

		\subsubsection{Configurations in the dynamic program}\label{sec:configs}
		
		We start with a simple bound on the number of crossings per portal. 		
		
		\begin{lemma}\label{lem:max_cross}
			In an optimal portal respecting solution, each portal is crossed at most $m+n$ times (by all $m$ paths in total).
		\end{lemma}
		\begin{proof}
			If a path crosses a portal twice then at least one input point is visited between these crossings since otherwise we can remove that part from the solution. There are $m$ paths and so there are at most $m+n$ crossings per portal.
		\end{proof}
		
		N.B. It would be easy to show (as in~\cite{Arora98JACM}) that any path crosses each portal at most twice. However, we will not use that property in the analysis and only use the easy upper bound on the total number of crossings per portal as given in the lemma above.
		
		Although paths are not directed by definition,  we shall assume that they are. Then, the $m$ directed paths can be seen as a  flow of value $m$ from $a'$ to $b'$ through a graph with the portals as its vertices and every input point is on one of these flow paths. 

		\begin{definition}
			Given a portal respecting, directed solution $\Gamma$, and a dissection square $S$, we define the \flowgraph\ for $S$ and $\Gamma$ as the directed graph with the portals of $S$ as vertices and an arc between any two portal-vertices $p$ and $p'$ if there is at least one path that enters $S$ in $p$ and leaves $S$ in $p'$. 
		\end{definition}
			\begin{figure}
			\centering
			\includegraphics[width=0.7\linewidth]{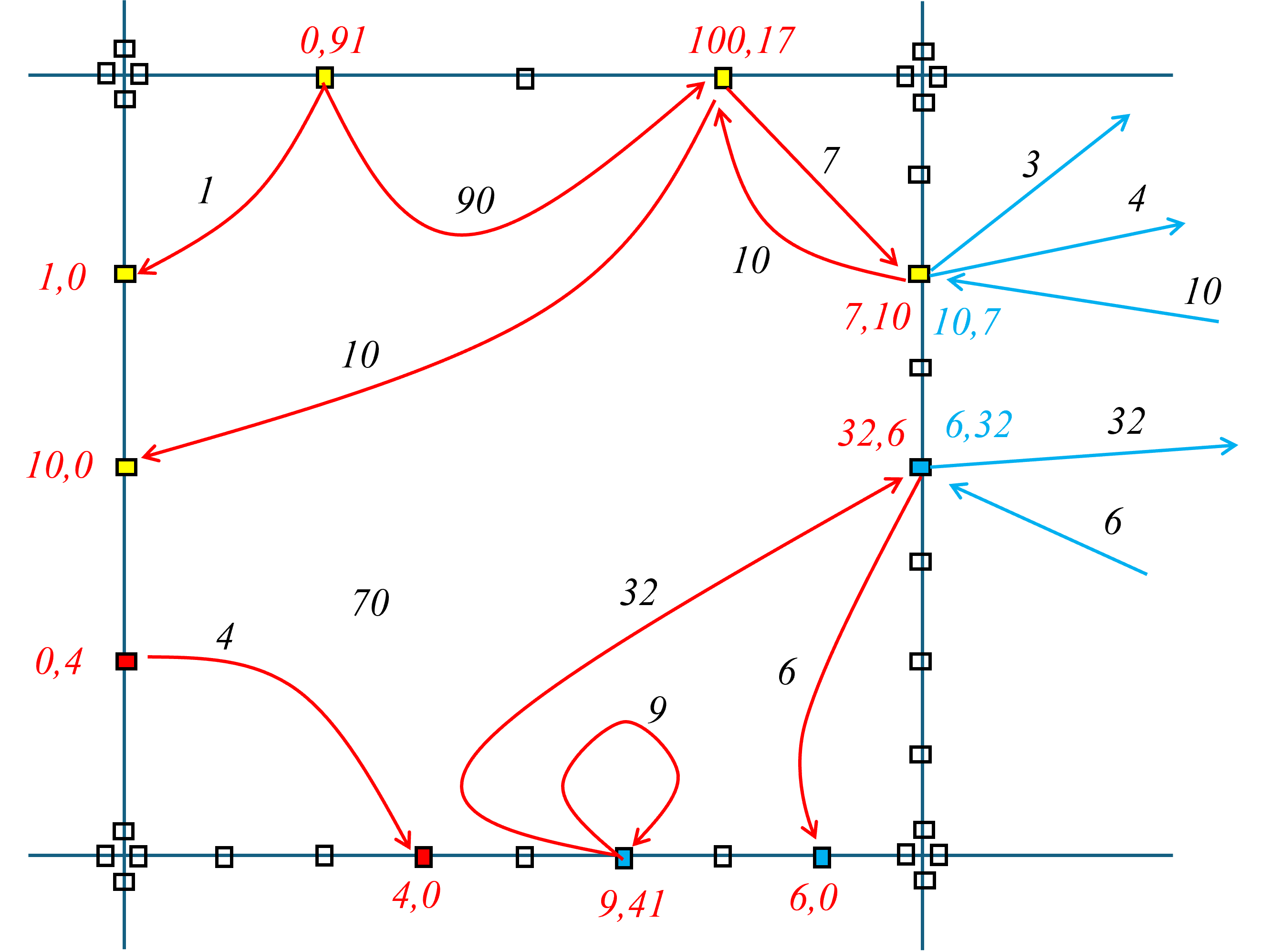}
			\caption{Example of a flow graph (the red arcs) with flow values on the arcs and the total in- and out flow on the portals. The in- and out flow on both sides of a portal should be opposite since the flow value of an arc is the value that goes through the portal.  }
			\label{fig:flowgraph}
		\end{figure}
		Note that we define the flow graph as unweighted, i.e., there are no flow values on the arcs, and that it could have parallel arcs in opposite directions. In the DP, we shall enumerate over all possible flow values on the arcs. (See Figure~\ref{fig:flowgraph}.) 
		
		Any flow graph has the nice property that it is outer planar. This holds since any feasible solution (portal respecting or not) can be modified without changing any part of the solution such that paths do not cross each other.  (See Figure~\ref{fig:uncross}.) Of course, for the uncrossing it is crucial that there are no capacity constraints on the individual paths and that paths have the same start point $a$ and end point $b$. Any simple, undirected outer planar graph on $K$ vertices has at most $2K-3$ edges. Our flow graph can have parallel arcs in opposite directions. So, the number of arcs is at most $4K-6$.   
	    In our setting $K$ is the number of portals in a dissection square: $K=\tilde{O}(\log n)$. 
			\begin{figure}[]
			\centering
			\includegraphics[width=0.6\linewidth]{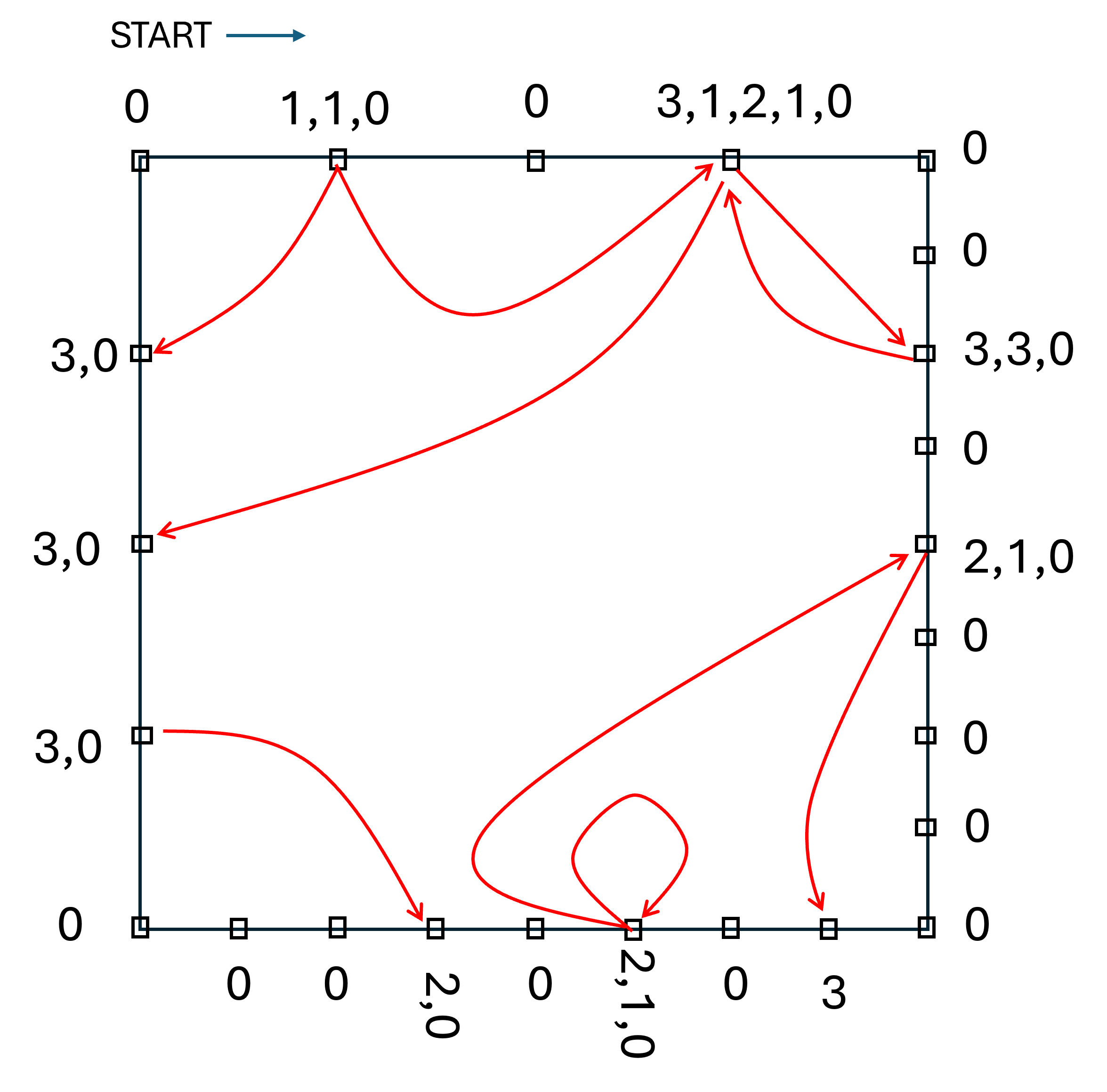}
			\caption{The flow graph is encoded by a sequence of numbers from $\{0,1,2,3\}$ with the following meaning: (0) move to the next portal, (1) a new outgoing arc, (2) a new incoming arc, (3) an arc that was encountered before. This graph corresponds with  
				the sequence:	$0110031210003300210000030210020000303030$. }
			\label{fig:flowgraphcount}
		\end{figure} 
			\begin{figure}
			\centering
			\includegraphics[width=0.65\linewidth]{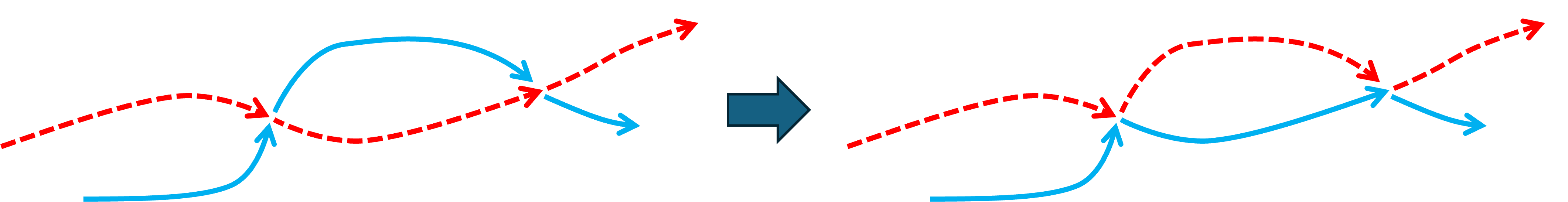}
			\caption{Since there is no capacity constraint and endpoint are the same, crossing paths can always be uncrossed at zero cost.}
			\label{fig:uncross}
		\end{figure} 
		\begin{lemma}\label{lem:nflowgraphs}
			For each dissection square, the number of flow graphs is $n^{\eeps}$.  
		\end{lemma}
		\begin{proof}
		Consider any dissection square and let $K=\tilde{O}(\log n)$ be the number of portals. Fix any portal to start a walk around the boundary. We create a sequence of length $O(K)$ where every entry takes one of 4 values with the following meaning: (1) a new outgoing arc, (2) a new incoming arc, (3) an arc that was encountered before, (0) move to the next portal. Any flow graph defines a unique sequence, and given a valid sequence, we can reconstruct the flow graph from it. 
		The length of any sequence is at most $K$ plus twice the number of arcs, which is at most $K+4K-6\le 5K$. Hence, the number of sequence is $4^{O(K)}=n^{\eeps}$. 
			An  example is given in Figure~\ref{fig:flowgraphcount}.    (Parallel arcs with opposite directions can be swapped, giving two encodings for the same graph. To make it unique we could for example always orient the two arcs clockwise, as in the figure.)     
		\end{proof}
		Note that for any feasible solution $\Gamma$ and dissection square $S$ we have \emph{flow conservation} on the boundary of $S$, i.e., the number of paths entering $S$ equals the number of paths leaving $S$. For $S=B_0$, the inflow and outflow are both exactly $m$. Similarly to~\cite{Arora98JACM}, we store the necessary boundary information in the dynamic program. For TSP, this is given by the number of crossings per portal and how these crossings connect to each other (the pairing). In our DP, we shall not track individual paths in the DP though. There is no need to and the running time would be too large. For each dissection square $S$ we store a set of configurations with the information as defined below.    \\
		 
		\begin{definition}\label{def:config}
			A \emph{configuration} for a dissection square $S$ is given by the lists:
			\begin{enumerate}
				\item[] $\flow(S)$, listing the in- and outflow for $S$ for each of the portals of $S$.  
				We say that $S$ satisfies \emph{flow conservation} if the total flow into $S$ equals the total flow out of $S$ and define this as the \emph{flow value} of $S$. 
				\item[] $\comp(S)$, listing the components inside $S$, i.e., for each component the set of portals of $S$ in it. Arc directions are ignored here. 
			\end{enumerate}
		\end{definition}
		
		\begin{lemma}\label{lem:nconfig}
			For each dissection square, the number of configuration is $n^{\tilde{O}(\log n)}$. 
		\end{lemma}
		\begin{proof}
			By Lemma~\ref{lem:nflowgraphs}, the number of \flowgraphs\ is $n^{\eeps}$.
			The number of arcs in a flow graph is $\tilde{O}(\log n)$ and by Lemma~\ref{lem:max_cross}, the in- and outflow of every portal is at most $m+n$. Hence, the number of possible lists for $\flow(S)$ is bounded by $n^{\tilde{O}(\log n)}$. For each flow graph, the component information is unique. Hence, the total   number of configuration is $n^{\tilde{O}(\log n)}$. 
		\end{proof}
		Before doing the actual DP, we first enumerate many configurations which are then checked for their validity in a bottom up way in the DP.  For the grid cells (level $\rho$) we store a value already in the process of the enumeration of the configurations. For the dissection squares of level $i<\rho$ we compute a value for the configurations in the dynamic program.

		For each grid cell in $B_0\setminus B_1$, we store only one configuration, namely the one that corresponds to the two  paths that connect the  anchor points by $m$ parallel paths to the boundary of $B_0$. It is good to realize that these paths can be taken arbitrarily since the cost for these is fixed. The only purpose is to have flow conservation on any dissection square by moving the source $a$ and sink $b$ to the boundary.
				
		 For the grid cells $S$ in $B_1$ we do the following. We enumerate all possible flow graphs and all possible flow values on the arcs of that graph. This gives an in- and outflow for each portal (and we automatically have flow conservation on $S$). We store the in- and outflows and the information of the different components  (as in Definition~\ref{def:config}).   
		 If the cell contains no point, then all connections are straight lines. If there is a point, then an optimal solution is easy to find. The cheapest option is to make all connections straight and then bend one path to serve the point. We declare all computed configurations on grid cells as \emph{valid} and store their corresponding solutions.
		 	
		 For all level  
		 $i$ squares with $1\le i<\rho$, we enumerate all possible flow graphs and all possible flow values on the arcs of that graph. This gives an in- and outflow for each portal (and we automatically have flow conservation on $S$). We store the in- and outflows and the information of the different components  (as in Definition~\ref{def:config}). Note that we do not store the flow graph itself in the DP table. In other words, the flow graphs are only used to enumerate candidate configurations, which we later check for validity.  
		
		For the level 0 square $B_0$, we only enumerate configurations where the in- and outflow is exactly $m$ and the flow graph is connected, i.e., there is only one component.

		\subsubsection{The Dynamic Program (DP)}
		The graph we need to construct in the DP is directed and the inflow equals the outflow on each portal, except for the two anchor points $a',b'$. In fact, we have something stronger: The flow into a portal $p$ on one side of the grid line is the same as flow out of $p$ on the other side of the grid line, and vice versa. 
		In the DP, we check for this flow conservation in a bottom up way and also check consistency of the component information and check that no isolated components are formed, i.e., components in a dissection square $S$ that have no connection to the boundary of $S$. If a configuration fulfils these checks then we declare it as valid and store the cost: $\cost(S,\flow ,\comp)$. 		
			
			\begin{definition}\label{def:validity}
				For a dissection square $S$ and lists $\flow$ and $\comp$, we say that 
				$(\flow,\comp)$ is a \emph{valid configuration} for $S$ if there are valid configurations $(\flow(S^{(i)}),\comp(S^{(i)}))$ for its 4 children  $S^{(i)}$ ($i \in \{1,2,3,4\}$) in the dissection tree such that:
				\begin{enumerate}
					\item The component information is consistent and no isolated components are formed.
					\item The flow information is consistent.  For a portal $p$  on the boundary of $S$, the flow values should be the same for $S$ and the child containing $p$. If $p$ is a portal on the boundary of two adjacent children of $S$ then the flows should be opposite for the two sides of the portal, i.e., the flow  continues through the portals. Note that flow conservation for the  4 children then implies flow conservation for $S$.   
				\end{enumerate}
				In case these properties hold, we say that the configuration of $S$ is \emph{consistent} with the configurations of its 4 children.  
			\end{definition}
			
			The validity check does not explicitly require flow conservation on $S$, but, since all valid configurations on the grid cells satisfy flow conservation, the flow conservation on dissection squares is inherited from the configurations on the grid cells.

			\begin{algorithm}
				\caption{The Dynamic Program:}\label{alg:DP1}
				\begin{algorithmic}
					\State \textbf{Input:}
					\begin{itemize}
						\item Valid configurations for all grid cells together with their cost.
						\item Candidate configurations for all other dissection squares.
					\end{itemize} 
					\For {$j=\rho-1$ to 0}:
					\For {all enumerated configurations $(S,\flow,\comp)$ of level $j$}:
					\State Check validity according to Definition~\ref{def:validity}. Remove if invalid.
					\State {The cost of $(S,\flow,\comp)$ is defined as the smallest total cost over 4 valid
					\State	 and consistent configurations of the children of $S$. Keep a \emph{pointer} from $S$ to
					\State those 4 configurations.}
					\EndFor
					\EndFor\\
					\Return the cheapest valid configuration of level $0$ and all the corresponding configurations of the grid cells (which can be derived from the pointers).   
					Denote this solution by $\Gamma_{\DP}$.
				\end{algorithmic}
			\end{algorithm}

			From the solution returned by the DP we remove the anchor  paths to get the solution for the $m$-paths instance. Remember that we defined $\OPT$ as the optimal cost among the portal respecting solutions for the $m$-paths problem with anchor points $a$ and $b$.
			\begin{lemma}
				The DP returns a solution consisting of $m$ paths from $a'$ to $b'$ of cost no more than $\OPT +Z^{0}$, where $Z^{0}$ is the cost of the anchor paths. Hence, deleting the anchor paths from the solution gives an optimal portal respecting solution for the $m$-paths problem. 
			\end{lemma}
			\begin{proof}
				A solution returned by the DP, consists of one solution for each grid cell, which is  a set of directed paths together with a flow value on each path. The grid cell solutions together form a flow from $a'$ to $b'$ of value $m$ with the property that the flow network is connected and every input point is visited. 
				
				On the other hand, let $\Gamma^*$ be an optimal portal respecting solution, including the anchor paths. The cost of $\Gamma^*$ is, by definition, $\OPT +Z^{0}$. 
				Solution $\Gamma^*$ defines a natural configuration for each dissection square $S$. 
				All these configurations were generated and verified as valid configurations. So the optimal cost for the DP is at most $\OPT +Z^{0}$. 
				\end{proof}

			\paragraph{Running time}
			In our algorithm, we first list many configurations and then check all of them for their validity in a DP in a bottom up way. 
			The depth of the dissection tree is $\tilde{O}(\log n)$ and the size (number of dissection squares) is $n^{\tilde{O}(1)}$. By Lemma~\ref{lem:nconfig}, the number of configurations for each dissection square is $n^{\tilde{O}(\log n)}$ and these can be enumerated in the same amount of time. Checking validity of any  configuration takes no more than $n^{\tilde{O}(\log n)}$. Hence, the total running time is $n^{\tilde{O}(\log n)}$.

			\subsection{The general $m$-paths problem.}\label{sec:gen_mpaths1}
			
			We defined the $m$-paths problem as the special case of the general $m$-paths problem where there is only one type crossing the square and it does so only once. In general, the number of types $t$ is $\eeps$ and each type $t$ crosses the square $\npairs^{(t)}=\eeps$ times. 
			We can afford to guess exactly the number of input points visited for each $t$, for each anchor point pair, since this give $n^{\eeps}$ options. When we do this, then for each of these guesses, we get a problem of the following form. 	Our notation here differs a bit from what we used before, just to keep it simple.  			

			\begin{itemize}
				\item We are given a square (as in the general $m$-paths problem) and a list of $\eeps$ anchor point pairs. Also, for each entry $i$ of an anchor point pair in the list, a number of paths $m_i$, and number of input points $n_i$ is given.
				(An anchor point pair may appear multiple times in the list.) 
				\item Goal is to find a mincost solution for the combined problem where we need to find  exactly $m_i$ paths for anchor point pair $i$ and these paths visit exactly $n_i$ input points and each input point is visited by exactly one path. 
			\end{itemize}
			So the general $m$-paths problem can be reduced to combinations of $\eeps$ $m$-paths problems. For any such a combination, the $\eeps$ subproblems can be solved almost independently in parallel. The only connection between these subproblems is that each gets the correct number of input points, but this is done in the base case of the DP: the grid cells. Let's go over it in more detail.  
			
						The random grid is defined exactly the same. Note that when input points are moved to the middle of grid cells, we can no longer merge these to a single input point. So, in the base problems in the grid cells, we need to specify how many points from the grid cell go to each of the anchor pairs. Since there are ${\eeps}$ anchor point pairs, there are $n^{\eeps}$ choices.   
			
			A configuration for a dissection square $S$ in the DP is now a combined configuration for each of the $\eeps$ anchor pair problems. In addition, we store the number of points visited in $S$ by each of them. Paths of different anchor points pairs may cross but that is absolutely fine. We say that a configuration on dissection square $S$ is valid if each of the separate configurations is valid. The running time of the DP is 
			\[\left(n^{\tilde{O}(\log n)}\right)^{\eeps}=n^{\tilde{O}(\log n)}.\]

			\section{An $n^{\tilde{O}(\log\log n)}$ time approximation scheme.}\label{sec:loglog}
			The approach here is exactly the same as in the previous section but now we use rounded numbers for the flow values. This, however, is easier said than done.  By rounding flow values  we lose flow conservation, which then needs to be corrected at every level of the DP. We want to maintain flow conservation on each dissection square, i.e., the total flow in equals the total flow out of the square, but  we also want to maintain flow conservation in each portal. We shortly discuss these two properties here. 
			
			To maintain flow conservation for each dissection square in the DP, we will only enumerate configurations with this property. For this, we use the same procedure as was done in Section~\ref{sec:logn}. That means, we first generate a flow graph, and use rounded flow values on the arcs. But when we add up these numbers at the portals, we do not round them. See the  
			example of Figure~\ref{fig:rounding}.   
			
			Maintaining flow conservation at portals is more involved. Two types of errors may occur: There can be a difference in flow value at a portal on either side of the grid line (,i.e., among two adjacent squares of the same level), or a difference in flow value at a portal among squares at the same side of the grid line but at different levels. Correcting the imbalance is easily done by adding paths to the solution and we charged the cost to the flow values. For this, we need a lower bound on the cost of any portal respecting solution (that comes out of the DP) in terms of the number of portal crossings. But such a bound does not exist directly. 
 
			We cannot use the (standard) randomization argument as in Section~\ref{sec:logn}, when we analyzed the cost of making an optimal tour portal respecting, since our modifications here are done to solutions obtained from the DP, i.e., \emph{after} the randomization step. To guarantee the desired lower bound, one could add a cost for each crossing artificially in the DP and run the DP with this extended cost. Here, we use another approach where we simply alternate the directions of portals (as we did before for anchor points) such that a path cannot enter and leave a dissection square at the same portal and, consequently, must stay in the square for at least the inter portal distance. Intuitively, this simple idea gives a useful lower bound and only adds a factor 2 to the cost of making a path portal respecting. Note however, that this is problematic for portals at the intersection of grid lines, since we may need both directions for the detours. (Also note that the inter portal distance might be much larger than 1, the grid distance.) 
			A clever hierarchical approach for alternating the portal directions does give the desired property though. We explain the details in Section~\ref{sec:RestrictedPortalLayout}.  
			With this layout, the expected detour for moving a crossing to a portal is small (Lemma~\ref{lem:layoutdetour}), and additionally it gives a lower bound on the cost of a solution inside a dissection square (Lemma~\ref{lem:stayinsquare}). 
			Arguably, this approach is more complicated than adding a cost for each crossing artificially in the DP, as mentioned above, but we believe that the special portal layout with the two lemmas are of general interest for TSP-like problems in the Euclidean plane.

			We shall go over the details for the  $m$-paths problem in Section~\ref{sec:mpathsloglog} and discuss the general form in Section~\ref{sec:mpaths2}.

\subsection{An $n^{\tilde{O}(\log\log n)}$ time algorithm for the $m$-paths problem}\label{sec:mpathsloglog}
			Before describing the restricted portal layout in this section, lets first look at how to choose the  grid. This is done in exactly the same way as in the previous section. Further, we fix two portal respecting paths from anchor points $a$ and $b$ to points $a'$ and $b'$ on the boundary of $B_0$ and define $Z^{0}$ as the total cost of these two paths (weighted by a factor $m$).  
			
			\begin{figure}
				\centering
				\includegraphics[width=0.8\linewidth]{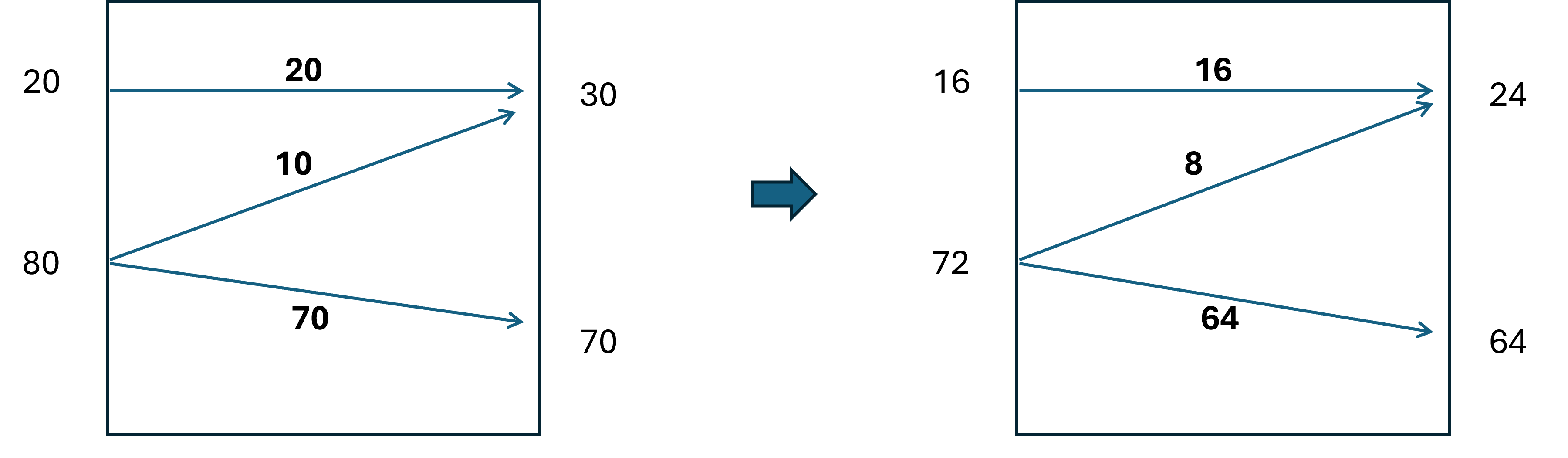}
				\caption{The rounding of flow values for a dissection square. We only round the flow values on the flow graph for the square, not at the portals. In the example, we round to powers of 2. Note that flow conservation on the square (flow in is flow out) is preserved by this rounding.}
				\label{fig:rounding}
			\end{figure}

			\subsubsection{A restricted portal layout}\label{sec:RestrictedPortalLayout}

			The sole purpose of the restricted portal layout is to have a lower bound on the cost of any solution within a dissection square $S$ in terms of the number of portal crossings of $S$. We use it to bound the cost for repairing the imbalance caused by using rounded flow values in the DP. 
			\bigskip 
			
			\noindent \textbf{Portal restriction procedure:}
			\begin{enumerate}
				\item Starting from the standard portal layout (as used in~\cite{Arora98JACM} and in Section~\ref{sec:logn})  we give an orientation to the 4 portals on each intersection of grid lines, exactly as shown in the upper left crossing of Figure~\ref{fig:portallayout}.
				\item For $i=1,2,..,\rho:$
				\begin{itemize}
					\item  Keep the 4 portals on any intersection of two level $i$ lines. Mark these portals.
					\item  For any level $i$ line, alternate the direction of all unmarked portals  of the level $i$ line. In particular, for the intersection of a level $i$ line with a level $j>i$ line, this means that one of the two directed crossings of the level $j$ line with the level $i$ line is removed. We keep and mark the remaining 3 crossings on each such intersection. See for example the horizontal level $i+1$ line in the figure.
				\end{itemize} 
			\end{enumerate}
			
			\begin{figure}
				\centering
				\includegraphics[width=0.99\linewidth]{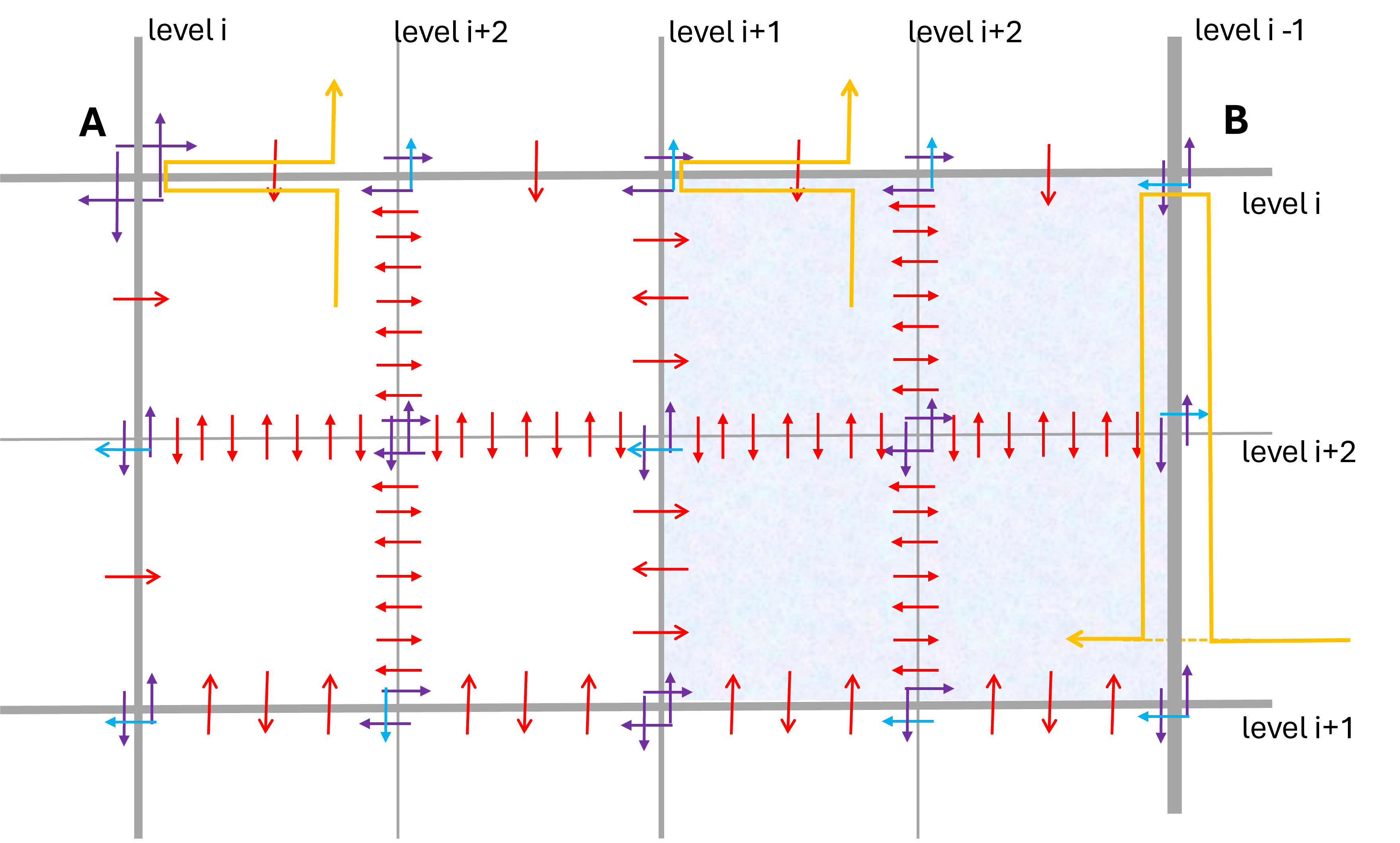}
				\caption{Impression of the restricted portal layout. (The portals which are drawn close to line intersections are actually at zero distance to the intersection.)    
					A \emph{restricted} portal respecting path only crosses the portals in the given direction. This special structure increases the expected detour for a crossing of a path with a grid line by only a factor 4 (compared to the standard layout) but has the additional property that the minimum path length inside a dissection square is at least the inter portal distance.
				}
				\label{fig:portallayout}
			\end{figure}
			
			\begin{lemma}\label{lem:layoutdetour}
				Any path can be made portal respecting such that the detour per crossing of a grid line is at most 4 times the inter portal distance of that line.
			\end{lemma}
			\begin{proof}
				Consider for example the horizontal level $i$ line in Figure~\ref{fig:portallayout}. Now consider the intersection of this line with 2 vertical lines that are of level $j\le i$ and that are consecutive. In the figure, these two intersections are marked $A$ and $B$. We have the following properties: (i)  At $A$ and $B$, we did not remove the portal crossings with the level $i$ line (the 2 vertical arrows at $A$ and the 2 vertical arrows at $B$), (ii) between $A$ and $B$, the portal crossings with vertical grid lines below the level $i$ line go left, and above the level $i$ line they go to the right, and (iii) between $A$ and $B$, the portals on the level $i$ line have alternating orientation. These 3 properties hold in general for the part of a level $i$ line between two consecutive intersections with higher level ($<i$) lines. Hence, for a any crossing with a grid line there always is a portal at distance at most twice the inter portal distance that can be used. 
			\end{proof}
			The lemma above states that we did not lose anything substantial in the analysis by using this portal structure. 
			We conclude that we can, from now on, restrict to portal respecting solutions using the special portal layout and denote by $\OPT$ the optimal cost among the portal respecting solutions. As before (see~\eqref{eq:exp_opt1}), we have 
			\[\text{Expected}(\OPT)\le (1+4\eps)\OPT_{\ses}.\] 
			The next lemma gives an interesting property and may be of of general interest for network optimization problems in $\RR^2$. We shall use it to prove Lemma~\ref{lem:realsolution}.
			\begin{lemma}\label{lem:stayinsquare}
				Let $S$ be a dissection square of level $i$. Any path that enters and leaves $S$ only at the restricted portals will stay in $S$ for at least the inter portal distance, $\eta_i$, of a level $i$ line.
			\end{lemma}
			\begin{proof}
				In the corners of a dissection square, a path can only leave or enter but not both. On a side between two corners, portals have alternating directions so any path must go between two different portals, and these are separated by at least the inter portal distance of a level $i$ line. (Note that the sides of a level $i$ dissection square are formed by two level $i$ lines, and two lines of higher level, which have a larger inter portal distance.)   
			\end{proof}

			\subsubsection{The DP with rounded flow values}
			
			We apply the same DP as in Section~\ref{sec:logn} but now with rounded flow values. The consistency check in the DP is replaced by a 'pseudo-consistency' check, which is formalized in Definition~\ref{def:validpseudo}.

			We will round to powers of $\alpha$, where $\alpha$ is specified later, and we define $N_{\FF}$ such that $\lfloor \alpha^{N_{\FF}}\rfloor=n$.  Let $\FF=\{0\}\cup \{\lfloor \alpha^i\rfloor \mid 0\le i\le N_{\FF}\}$. It turns out that we can choose $\alpha$ such that 
			$N_{\FF}=\tilde{O}(\log^3 n)$.
			
			Since we gave an orientation to every portal, all paths crossing a portal $p$, do so  in the same direction: The flow through  $p$ either goes into or out of $S$, but not both. We denote by $\flow(p,S)\ge 0$ the total flow  entering or leaving $S$ through $p$. The orientation of the portal tells what the direction of the flow is. 
			If dissection squares $S$ and $S'$ share a portal $p$, then in any feasible solution, $\flow(p,S)=\flow(p,S')$. However, when we run the DP with rounded numbers then we cannot maintain this identity and flow values for the same portal may differ among various dissection squares. We will show that we can repair this imbalance at a cost of a factor $\alpha$ at every level.  
			
			\paragraph{Enumeration of configurations}
			As before, a configuration for a dissection square is given by the flow information for the portals and the component information.
			For the enumeration of configurations we start, as before, from the flow graphs. So, for a dissection square $S$, we first generate all possible flow graphs on the portals of $S$. Then, for any flow graph we enumerate over all possible flow values, where we now restrict to values from $\FF$. Then, we add up these numbers at the portals to get the flow values for the portals. But we \underline{do not round} these numbers. (See Figure~\ref{fig:rounding} for an example.) Doing the rounding for the flow on the arcs and not for the portals ensures flow conservation on $S$. Also, the flow graph provides the necessary component information. Note that a flow value of 1 is not rounded down to 0, since $1\in\FF$. Hence, rounding has no effect on the component information.  
			Again, we do not store the flow graphs in the DP but only the configurations derived from them.

			The enumeration of configurations for the grid cells and square $B_0$ requires more attention. For the grid cells, we only enumerate configurations that are valid and optimal. This works almost the same as in Section~\ref{sec:logn}. 
			We try all possible flow graphs and use \emph{rounded} flow values. If the cell contains part of the anchor paths $(a',a)$ and $(b,b')$, then that part is a fixed part of the flow graph that we enumerate for that cell and we also round the flow value $m$. 
			The rounded values are added to obtain the flows at the portals, but these values are not rounded. If the square contains no input point then all connections are straight lines. If there is a point then the cheapest option is to make all connections straight and then bend one of the lines to serve the point. All these computed configurations are declared \emph{valid}.
			
			For $B_0$ we only want configurations with an in-and outflow that is not much larger than $m$. It turns out (Lemma~\ref{lem:costDP}) that a value of at most $\alpha^{\rho}m$ suffices. We put that as a hard constraint in the DP.  

			\paragraph{The Dynamic Program}

			For a dissection square $S$ we define the \emph{outer portals} as the portals on the boundary of $S$ and the \emph{inner portals} as the portals that are on two adjacent children of $S$. 
			
			\begin{definition}\label{def:validpseudo}
				For a dissection square $S$ and lists $\flow(S)$ and $\comp(S)$, we say that 
				$(S,\flow(S),\comp(S))$ is a \emph{valid configuration} if there are valid configurations $(S^{(i)},\flow(S^{(i)}),\comp(S^{(i)}))$ for its 4 children 
				$S^{(i)}$ ($i \in \{1,2,3,4\}$) in the dissection tree such that:
				\begin{enumerate}
					\item The component information is consistent and no isolated components are formed.
					\item The flow information is \emph{pseudo-consistent}. That means,  for an inner portal, which appears on adjacent children $S_i$ and $S_j$, we require that
					\begin{equation}\label{eq:pseudo1}
						\frac{1}{\alpha}\le \frac{\flow(p,S_i)}{\flow(p,S_j)}\le \alpha,
					\end{equation}  
					and for an outer portal, which is on $S$ and on one of its children $S_i$,  we require 
					\begin{equation}\label{eq:pseudo2}
						1\le \frac{\flow(p,S)}{\flow(p,S_i)}\le \alpha^2.
					\end{equation}		
				\end{enumerate}
				We say that the valid configuration is \emph{pseudo-consistent} with configurations of its 4 children, as defined above.   
			\end{definition}
			The reason for the difference in bounds between~\eqref{eq:pseudo1} 
			and~\eqref{eq:pseudo2} is that we will correct for the mismatch by adding paths such that the flow of any outer portal $p$ is exactly $\flow(p,S)$. So we want $\flow(p,S)\ge \flow(p,S_i)$. Consequently, the loss due to rounding is a factor $\alpha$ at every level.

			\bigskip 
			
			\begin{algorithm}
				\caption{The Dynamic Program:}\label{alg:cap}
				\begin{algorithmic}
					\State \textbf{Input:}
				\begin{itemize}
					\item Valid configurations for all grid cells together with their cost.
					\item Candidate configurations for all other dissection squares.
				\end{itemize} 
				\For {$j=\rho-1$ to 0}:
					\For {all enumerated configurations $(S,\flow,\comp)$ of level $j$}:
					\State Check validity according to Definition~\ref{def:validpseudo}. Remove if invalid.
					\State The cost of $(S,\flow,\comp)$ is defined as the smallest sum over 4 valid 
					\State and pseudo-consistent configurations of the children of $S$. Keep a \emph{pointer}  
					\State from $S$ to those 4 configurations.    	  	
					\EndFor
					\EndFor\\
					\Return the cheapest valid configuration of level $0$ and all the configurations of each level from which it was derived, i.e, one configuration for each dissection square of the quad tree. 
					Denote this by $\Gamma_{\DP}$.
				\end{algorithmic}
			\end{algorithm}

			The DP returns an answer which does, in general, not satisfy flow conservation. We call this a \emph{pseudo-solution}. In fact, $\Gamma_{\DP}$ gives one valid configuration, say $\Gamma_{\DP}(S)$ for every dissection square $S$. The cost of $\Gamma_{\DP}$ is simply the sum of the cost of all configurations on the grid cells. 
			Let $\mathcal{C}$ be the set of all grid cells and for $C\in \mathcal{C}$, let $\cost(C,\Gamma_{\DP})$ denote the cost of $C$ in pseudo solution $\Gamma_{\DP}$. In general, denote $\cost(S,\Gamma_{\DP})=\sum_{C\in S} \cost(C,\Gamma_{\DP})$. In particular, the cost of the solution  $\Gamma_{\DP}$ returned by the DP is  $\sum_{C\in \mathcal{C}} \cost(C,\Gamma_{\DP})$.

			Remember that we defined $\OPT$ as the optimal cost for a portal respecting solution without the fixed anchor paths, and $Z^{0}$ is the cost for the fixed anchor paths.

			\begin{lemma}\label{lem:costDP}
				The DP will always return a pseudo-solution $\Gamma_{\DP}$ of cost $Z_{\DP}$ with $Z_{\DP}\le \OPT +Z^{0}$.  
				Further, the flow value of the configuration for $B_0$ in $\Gamma_{\DP}$ is at most  $\alpha^{\rho }m$.
			\end{lemma}
			\begin{proof}
				We did put  $\alpha^{\rho }m$ as the maximum flow value of $B_0$ as a hard restriction in the DP, so it is enough to show the first part of the claim, namely, that at least one outcome of the DP exists with a cost of at most $\OPT +Z^{0}$.
				
				Consider an optimal  portal respecting solution $\Gamma^*$, including the anchor paths. The cost of this solution is  
				$\OPT +Z^{0}$.
				The solution $\Gamma^*$ defines a natural configuration for each dissection square $S$ in the following way: First, consider the flow graph that has an arc $p,p'$ if there is a path in $\Gamma^*$ that enters $S$ in $p$ and leaves $S$ in $p'$. For every arc in the flow graph this gives a flow value which is then rounded down to the nearest number in $\FF$. These rounded numbers are then added to obtain the flow values for the portals. (These sums are not rounded.) The components of the flow graph together with the flow values for the portals define the configuration. Denote this configuration by ${\Gamma^*}(S)$. 
				
				The configuration ${\Gamma^*}(S)$ on a dissection square $S$ will, in general, not be a valid configuration in our DP due to rounding errors. However, it is easy to show that when we multiply all flow values by the right power of $\alpha$ then the resulting configuration is in the DP table. For any dissection square $S$ of level $\rho-i$ and $i\in\{1,2,\dots,\rho\}$,  denote by $\correct(\Gamma^{*}(S))$ the configuration ${\Gamma^*}(S)$ where we multiply all flow values by a factor $\alpha^{i}$.
				
				We prove by induction on the level that $\correct(\Gamma^{*}(S))$ is in the DP table for all $S$. 
				For $i=0$, we are at level $\rho-0=\rho$, which are the grid cells. We enumerated all possible rounded configurations, and these include those of optimal solution $\Gamma^*$. For $i=0$, $\correct(\Gamma^{*}(S))$ is equal to ${\Gamma^*}(S)$ so this is covered in the DP table.

				Now assume that $\correct(\Gamma^*(S'))$ is in the DP table for all squares $S'$ of level $\rho,\rho-1,\dots, \rho-(i-1)$. Let $S$ be a dissection square of level $\rho-i$. We prove that 
				$\correct(\Gamma^*(S))$ is pseudo-consistent with the configurations  
				$\correct(\Gamma^*(S^{(h)}))$, $h\in\{1,2,3,4\}$, of its 4 children, and therefore is a valid configuration by Definition~\ref{def:validpseudo}.
				
				First note that the component information of 	$\correct(\Gamma^*(S))$ is consistent with the configurations  
				$\correct(\Gamma^*(S^{(h)}))$ of its children. This follows from the fact that rounding has no effect on the flow graph since a flow value of 1 is never rounded to zero ($1\in\FF$).     
				
				Let $p$ be a portal on the boundary of two children of $S$, say $S^{(1)}$ and $S^{(2)}$, which are of level $\rho-(i-1)$. Let  $\flow^*(p)$ be the true, unrounded flow through $p$ in $\Gamma^*$. Let $\flow(p,S^{(1)})$ be the rounded value as defined by $\correct(\Gamma^{*}(S^{(1)})$. Then 
				\[\alpha^{i-2}\flow^*(p)\le \flow(p,S^{(1)}) \le \alpha^{i-1}\flow^*(p).\]
				Similarly 
				\[\alpha^{i-2}\flow^*(p)\le \flow(p,S^{(2)}) \le \alpha^{i-1}\flow^*(p).\]
				Hence, 
				\[1/\alpha\le \flow(p,S^{(1)})/ \flow(p,S^{(2)})\le \alpha.\]
				Hence, constraint~\eqref{eq:pseudo1} is satisfied. 
				
				Now assume $p'$ is a portal on a boundary shared by $S$ and $S^{(1)}$.  
				Again, let  $\flow^*(p')$ be the true, unrounded flow through $p'$ in $\Gamma^*$. Let $\flow(p',S^{(1)})$ be the rounded value as defined by $\correct(\Gamma^{*}(S^{(1)})$. Then 
				\[\flow(p',S^{(1)})\ge \alpha^{i-2} \flow^*(p',S^{(1)}), \]
				and 		
				\[\flow(p',S)\le \alpha^{i} \flow^*(p',S). \]
				Hence,   
				\[1\le \flow(p',S)/\flow(p',S^{(1)}\le \alpha^2.\]
				So, Constraint~\eqref{eq:pseudo2} is satisfied as well and the configuration $\correct(\Gamma^*(S))$ is taken as a valid configuration in the DP table.
				Note that $\correct(\Gamma^*(B_0))$ has flow value exactly $\alpha^{\rho}m$.
				
			\end{proof}
			
			For any valid configuration of a dissection square $S$ in the DP we define a \emph{real solution} as a set of directed paths in $S$ that match the configuration. This is formally defined below.  
			
			\begin{definition}\label{def:realsolution}
				Let $(\flow,\comp)$ be a valid configuration for dissection square $S$. We define a \emph{real} solution for this configuration as a set of directed paths in $S$ that:
				\begin{itemize}
					\item together visit all points inside $S$
					\item have start- and the endpoint at the portals of $S$
					\item match the flow values \flow. More precisely: For an outgoing portal $p$ the flow value equals the number of paths that have $p$ as endpoint. For an incoming portal $p$, the flow value equals the number of paths that have $p$ as starting point.
					\item satisfy $\comp$. More precisely: If two portals are in the same component according to $\comp$ then they are also connected by the set of paths (ignoring the directions). 
				\end{itemize}
		The real solution as defined above does not need to be portal respecting inside $S$, it only needs to fulfill the boundary information $(\flow,\comp)$. 
			\end{definition}
			 
			Now we define $\alpha$ such that 
			\[(\alpha^2-1)= \frac{\eps\ln (1+\eps^2)}{8\rho^2}.
			\] 
			Remember that $\rho= \lfloor\log_2 (n/\epsilon)\rfloor$. So clearly $1<\alpha<2$. We use this equality for $\alpha^2-1$ in Equation~\ref{eq:repair_bound}. 
			To solve $\lfloor \alpha^{N_{\FF}}\rfloor=n$, note that 
			\[\alpha^{N_{\FF}}=(\alpha^2)^{N_{\FF}/2} = \left(1+\frac{\eps\ln(1+\eps^2)}{8\lfloor\log_2 (n/\epsilon)\rfloor^2} \right)^{N_{\FF}/2},\]
			 which implies $N_{\FF}=\tilde{O}(\log^3 n)$.

			\begin{lemma}\label{lem:realsolution}
				Given an optimal  pseudo-solution $\Gamma_{\DP}$ of cost $Z_{\DP}$  returned by DP, we can add arcs to obtain a set $\Gamma'$ consisting of $m'$-paths between $a'$ and $b'$ visiting all points in $\ses$ of cost at most $ (1+\eps^2)Z_{\DP}\le (1+\eps^2)(\OPT+Z^{0})$, and $m\le m'\le \alpha^{\rho}m\le (1+\eps^2)m$. 
			\end{lemma}
			\begin{proof}
				We show how to turn the pseudo-solution into a solution $\Gamma'$ by only adding arcs and not making other changes. 
 
				Intuitively, since the imbalance is small, we should be able to correct it at low cost. We will correct the flow level by level in a bottom up way and charge the cost to the number of portal crossings at the current level and use Lemma~\ref{lem:stayinsquare}.

				For the grid cells, the valid configurations were defined from real solutions so we do not need to add anything.
				Now let $j\ge 1$ and let $S$ be a dissection square of level level $i=\rho-j$ and let $\Gamma_{\DP}(S)=(\flow,\comp)$ be the valid configuration for $S$ as defined in $\Gamma_{\DP}$ and let  $(\flow^{(h)},\comp^{(h)})$ be the corresponding valid configurations for its 4 children $S^{(h)}$, $h\in\{1,2,3,4\}$ that are pseudo-consistent with  
				$(\flow,\comp)$. Assume that we already constructed real solutions, say $\Gamma'(S^{(h)})$, for these 4 configurations, as defined in Definition~\ref{def:realsolution}.  
				Let $\cost(S^{(h)})$ be the cost for each of these solutions. Now we put these 4 solutions together and add arcs.
				
				Let $N^{(h)}$ be the flow value of $\flow^{(h)}$, i.e., $2N^{(h)}$ is the number of portal crossings.
					We know use Lemma~\ref{lem:stayinsquare} for the real solutions $\Gamma'(S^{(h)})$. 
				\begin{equation}\label{eq:portalcost}
					\cost(S^{(h)})\ge N^{(h)}\eta_{i+1}, \text{ for }i\in\{1,2,3,4\}.
				\end{equation}

				By constraint~\eqref{eq:pseudo2} we can solve the imbalance at the inner and outer portals by adding arcs such that:
				\begin{itemize}
					\item[(i)] the flow values at the outer portals match that of $\flow$. Hence, we have flow conservation on $S$.
					\item[(ii)] we have flow conservation on each of the inner portals 
				\end{itemize}
				To see how this can be done, pick an arbitrary new point, say $w\in \RR^2$, inside $S$ and add directed paths between $w$ and the inner and outer portals. Since we have flow conservation on the inner portals and flow conservation on $S$, we have flow conservation on $w$. For an example see Figure~\ref{fig:balance}. N.B., the directed paths that we add are not portal respecting but simply straight lines in $\RR^2$.  A real solution as defined in Definition~\ref{def:realsolution} only needs to be portal respecting on the portals on the outer portals of $S$.    
				\begin{figure}
					\centering
					\includegraphics[width=0.65\linewidth]{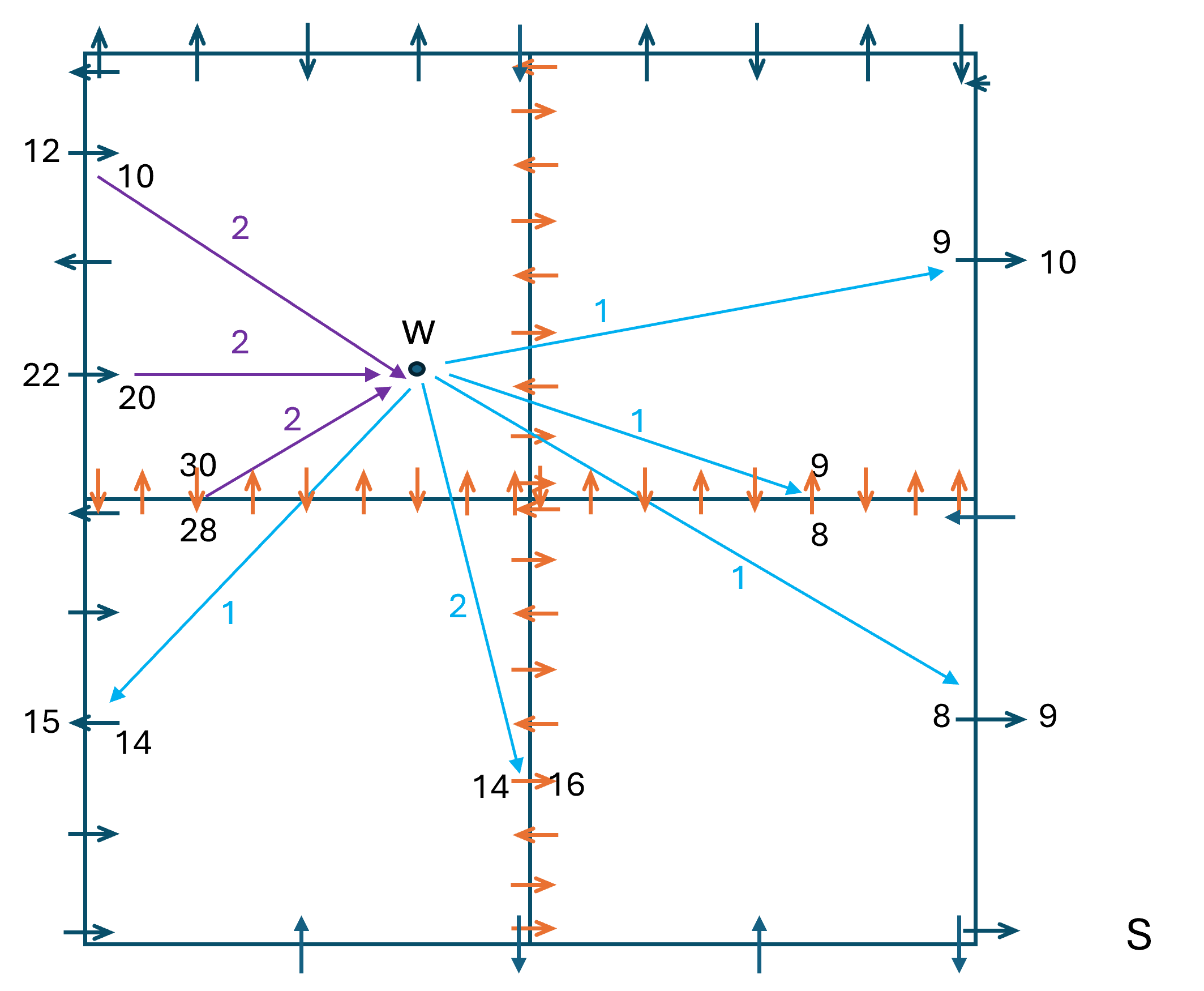}
					\caption{Repairing flow conservation. Without point $w$, the large square $S$ and its 4 children each satisfy flow conservation separately. For example, the upper left square has an in and out flow of 30, and $S$ has flow value 34. We repair the mismatch by adding one arbitrary point $w$ and arcs between $w$ and all portals with a mismatch.
 	The stored solutions on the 4 children together with point $w$ give a solution for $S$ with flow value $34$, which is equal to the flow value of $S$ before repairing. 
						N.B. The numbers outside the large square are for $S$ and each number inside the square $S$ is for one of the children. The figure only shows 2 levels. In general, a pseudo solution contains $\tilde{O}(\log n)$ flow values for each portal: 2 for every level.}
					\label{fig:balance}
				\end{figure}
				
				By definition of pseudo consistent, the number of arcs we need to add is no more than $(\alpha^2-1)(N^{(1)}+N^{(2)}+N^{(3)}+N^{(4)})$, and each arc has length at most $L_i$ if we take $w$ close to the middle of $S$. ($L_i$ is the side length of a level $i$ square.) From Equations~\eqref{eq:eta},\eqref{eq:Li}, and~\eqref{eq:portalcost}, the total added cost is at most 
				\begin{eqnarray*}
					&&(\alpha^2-1)L_i\sum_{h=1}^4 N^{(h)}\\
					&\le& (\alpha^2-1)L_i/\eta_{i+1} \sum_{h=1}^4 \cost(S^{(h)})\\
					&\le& (\alpha^2-1)(8\rho/\eps) \sum_{h=1}^4 \cost(S^{(h)}).
				\end{eqnarray*}
				When we do this repairing for all level $i$ squares, then the total cost increases by at most a factor  $1+(\alpha^2-1)8\rho/\eps$. So, over all $\rho$ levels, the cost of the pseudo solution of the DP increases by at most a factor   
				\begin{equation}\label{eq:repair_bound}
					\left(1+(\alpha^2-1)\frac{8\rho}{\eps}\right)^{\rho}=\left(1+\frac{\ln(1+\eps^2)}{\rho}\right)^{\rho}< e^{\ln(1+\eps^2)}=1+\eps^2.
				\end{equation}
				This completes the prove for the cost. To see that the number of paths is at most $(1+\epsilon)m$ observe that the flow value of a dissection square $S$ does not increase due to the repairing. The flow value of an outer portal of $S$ in the constructed real solution equals the flow value of $S$ in the pseudo solution derived from the DP. Hence, by Lemma~\ref{lem:costDP} the number of paths $m'$ is at most $\alpha^{\rho }m$.
				Now we inequality~\eqref{eq:repair_bound} to bound $\alpha^{\rho}$.
				\begin{eqnarray*}\label{eq:alpha^rho}
					\alpha^{\rho} <\alpha^{2\rho}= \left(1+(\alpha^2-1)\right)^{\rho}<\left(1+(\alpha^2-1)\frac{8\rho}{\eps}\right)^{\rho}< 1+\eps^2.
				\end{eqnarray*}

			\end{proof}
		Note that $\Gamma'$ was obtained by adding arcs to pseudo-solution $\Gamma_{DP}$. This pseudo-solution is composed of a solution for each of the grid cells, which we computed at the beginning of the DP, and which together do not satisfy flow conservation in general. This was repaired in the lemma above by only adding arcs. It is important to realize (and we use that in the proof of the next lemma) that the computed real solution still contains the two anchor paths.  
		We defined the cost of these paths as $Z^{0}$, which is no more than $mL_0$ plus the cost of making them portal respecting. For this, however, we can't use the expected detour argument since the paths are fixed in $B_0$. So lets take a rough upper bound on the cost of the two anchor paths, which we shall use in the proof for the next lemma. 
			\begin{equation}\label{eq:boundZ_0}
				Z^{0}\le 2mL_0.   
			\end{equation}
			\begin{lemma}\label{lem:removeanchorpaths}
				Given solution $\Gamma'$ as described in Lemma~\ref{lem:realsolution}, we can find a solution to the $m$-paths instance of cost at most $(1+14\eps)\OPT$.
			\end{lemma}
			\begin{proof}
				For a bound on $\OPT$ we go all the way back to Equation~\eqref{eq:LB_OPTpaths}. This simple bound was derived directly from the fact that the anchor points $a$ and $b$ are different. Hence, 
				\begin{equation}\label{eq:boundOPT}
					\OPT\ge L_1m\eps/2=L_0m\eps/4. 
				\end{equation}	
				Solution $\Gamma''$ is obtained from $\Gamma'$ by
				\begin{itemize}
					\item deleting the $m$ anchor paths from $a'$ to $a$,
					\item deleting the $m$ anchor paths from $b$ to $b'$,
					\item adding $m'-m$ paths from $a$ to $a'$, and 
					\item adding $m'-m$ paths from $b'$ to $b$. 
				\end{itemize}
				Then $\Gamma''$ is a set of $m'$ paths from $a$ to $b$. 
				Using Lemma~\ref{lem:realsolution} and equation~\ref{eq:alpha^rho}, we conclude that the cost of $\Gamma''$ is at most 
				\begin{eqnarray*}
					&&(1+\eps^2)(\OPT+Z^{0})-Z^{0}+(m'-m)L_0\\
					&=&(1+\eps^2)\OPT+\eps^2 Z^{0}+(m'-m)L_0\\
					&\le &(1+\eps^2)\OPT+2\eps^2 mL_0+(\alpha^{\rho}-1)mL_0\\
					&\le &(1+\eps^2)\OPT+3\eps^2 mL_0\\
					&\le  &(1+\eps^2)\OPT+12\eps\OPT\\
					&\le  &(1+13\eps)\OPT.
				\end{eqnarray*}
				\begin{figure}
					\centering
					\includegraphics[width=0.5\linewidth]{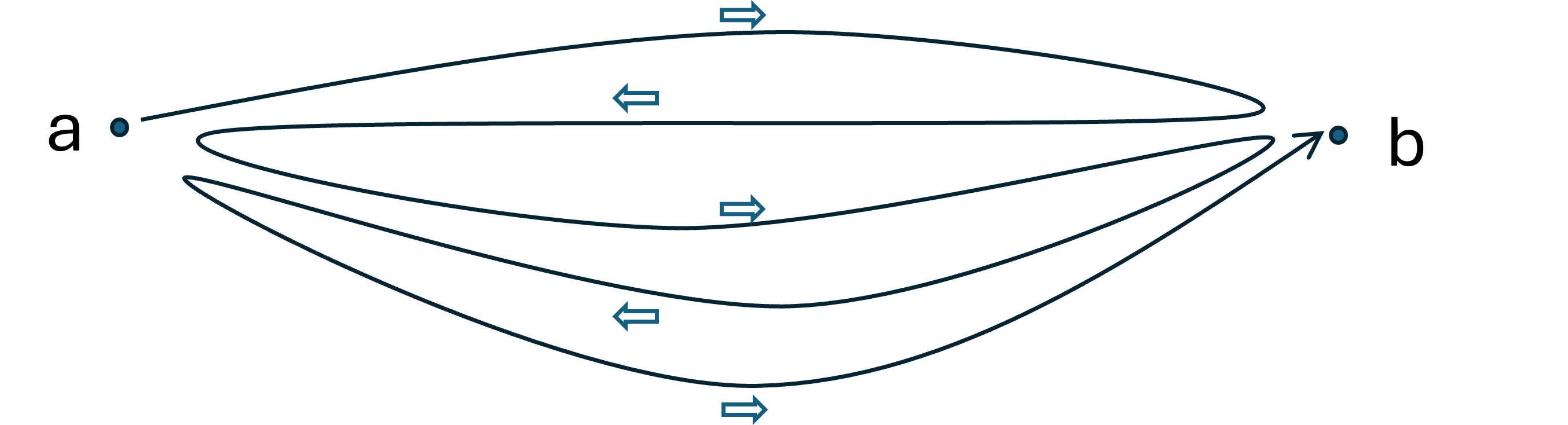}
					\caption{Reducing the number of paths from 5 to 1. }
					\label{fig:Pathreduce}
				\end{figure}
				Finally, we need to bring down the number of paths from $m'$ to $m$. This is completely trivial. See Figure~\ref{fig:Pathreduce} for an example. 
				If  $m'-m$ is an even number then we can artificially reduce to number of paths to $m$ at no cost by simply changing the direction of $(m'-m)/2$ of the paths. 
				If $m'-m$ is an odd number then we change the direction of $(m'-m+1)/2$ of the paths and add one direct connection from $a$ to $b$. Hence, the extra cost for reducing the number of paths to $m$ is no more than 
				$d(a,b)\le OPT/m<\eps\OPT$.		(Here we use that $m>1/\eps$ as discussed in the beginning of Section~\ref{sec:logn}.)
				We conclude that the total cost for the $m$-paths solution is no more than $(1+13\eps)\OPT+\eps\OPT=(1+14\eps)\OPT$. 
			\end{proof}

			\paragraph{Running time}
			First, lets bound the size of the DP table. The number of dissection squares is $O(4^{\rho})=O(4^{\log(n/\eps)})=n^{\eeps}$. For each dissection square, the number of possible flow graphs is $n^{\eeps}$. The size of $\FF$ is $\tilde{O}((\log n)^3)$ and the number of arcs in a flow graph is $\tilde{O}(\log n)$. So, the number of possible ways to put flow values on arcs (and consequently, the portals) is 
			\[\tilde{O}((\log n)^3)^{\tilde{O}(\log n)}=(\log n)^{\tilde{O}(\log n)}= n^{\tilde{O}(\log \log n)}.\] 
			Hence, the size of the DP table is $n^{\tilde{O}(\log \log n)}$. 
			Verifying validity of any entry in the table can be done in 
			$n^{\tilde{O}(\log \log n)}$ time so the total time for the DP is $n^{\tilde{O}(\log \log n)}$. 
			Finally, given an optimal  pseudo solution by the DP, we can turn it into a real solution in 
			$n^{\eeps}$ time since we only need to add arcs in each dissection square in a straightforward way.
			
			\subsection{The general $m$-paths problem.}\label{sec:mpaths2}
			The analysis of the general $m$-paths problem from Section~\ref{sec:gen_mpaths1} applies here as well.
			An instance of the general $m$-paths problem is in fact a set of $\eeps$ $m$-paths instances that need to be solved in parallel. We refer to these simply as \emph{the instances} and label them by $i=1,2,\dots$.  
			The DP can be done almost independently for the instances. The only connection is in the base case where the input points are divided over the different instances.     
						The running time of the DP is 
			\[\left(n^{\tilde{O}(\log\log n)}\right)^{\eeps}=n^{\tilde{O}(\log\log n)}.\]
			
			The only important difference is that the DP no longer gives an exact solution but a pseudo-solution that needs to be repaired. The repairing can be done independently for each instance and only takes $n^{\eeps}$ time per instance as discussed above. 
			
			It remains to bound the cost for repairing. In our analysis of the cost of repairing we used that $m>1/\eps$. Note however, that for an $m$-paths instance $i$ with $m_i\le 1/\eps$, rounding of flow values is not needed. So, in the combined DP, we do not round the flow values for such an instance and, consequently, no repairing is needed.            
			To see that  the cost for repairing is not too large, note that the pseudo solution returned by the DP is actually a pseudo solution for each of the $m$-paths instances and the cost of the pseudo solution can be divided over the different $m$-paths instances accordingly. Also, an optimal solution of cost $\OPT$ for the combined problem is composed of a solution for each of the $m$-paths instances. Let $\OPT_i$ be the cost of the part for $m$-paths instance $i$. Following the analysis of the $m$-paths problem above for each of the instances separately, we find a solution for each instance $i$ of cost no more than $(1+14\eps)\OPT_i$ and the total cost is no more than   $(1+14\eps)\OPT$. 
			
			\section{Generalizations}
			We presented a  Q-PTAS for CVRP in the Euclidean plane but the reductions we used are more general.   
			The reduction to bounded instances holds for Euclidean spaces of any dimension $d$. The reduction to the $m$-paths problem is only polynomial for  Euclidean spaces of a fixed dimension $d$. The Q-PTAS for the  
			$m$-paths problem was presented for the plane only. Here, we used the fact that paths do not cross, to bound the number of possible flow graphs for any dissection square (Lemma~\ref{lem:nflowgraphs}). This no longer works so easily for higher dimensions.  
			
			Although our CVRP algorithm has quasi polynomial running  time and only works for the Euclidean plane, our reductions show that any PTAS for the $m$-paths problem in $\RR^d$, for some fixed $d$, will most likely yield a PTAS for CVRP in $\RR^d$. Arguably, a PTAS for the $m$-paths problem is easier to derive than a PTAS for CVRP.     
			 
			Adamaszek et al.~\cite{AdamaszekCL2010} showed that the reduction to bounded instances, together with the $n^{\log^{O(1/\epsilon)}n}$-time approximation scheme by Das and Mathieu~\cite{DasMathieu2014} implies a polynomial time approximation schemes for all $c<2^{\log^{\delta}n}$, for some $\delta(\epsilon)<1$. If we now use the $n^{\tilde{O}(\log\log n)}$-time approximation scheme, we see that a PTAS for CVRP exists for $c=n^{O(1/\log\log n)}$.

			\bibliography{VRPbib.bib}

@article{Haimovich1985bounds,
	title = {Bounds and heuristics for capacitated routing problems},
	author = {Haimovich, Mordecai and Rinnooy Kan, Alexander},
	journal = {Mathematics of Operations Research},
	volume = {10},
	number = {4},
	pages = {527--542},
	year = {1985},
	publisher = {INFORMS},
	doi = {10.1287/moor.10.4.527}
}

@article{Baker1994,
	author = {Baker, Brenda S.},
	title = {Approximation algorithms for {NP}-complete problems on planar graphs},
	journal = {Journal of the ACM},
	volume = {41},
	number = {1},
	pages = {153--180},
	year = {1994},
	doi = {10.1145/176}
}

@inproceedings{AsanoKTT97,
	author = {Asano, Tetsuo and Katoh, Naoki and Tamaki, Hisao and Tokuyama, Takeshi},
	title = {Covering points in the plane by $k$-tours: towards a polynomial time approximation scheme for general $k$},
	year = {1997},
	isbn = {0897918886},
	url = {https://doi.org/10.1145/258533.258602},
	doi = {10.1145/258533.258602},
	booktitle = {Proceedings of the Twenty-Ninth Annual ACM Symposium on Theory of Computing},
	pages = {275–283},
	numpages = {9},
	location = {El Paso, Texas, USA},
	series = {STOC '97}
}

@article{Arora98JACM,
	author    = {Sanjeev Arora},
	title     = {Polynomial Time Approximation Schemes for {E}uclidean {TSP} and other Geometric Problems},
	journal   = {Journal of the ACM},
	volume    = {45},
	number    = {5},
	pages     = {753--782},
	year      = {1998},
	publisher = {ACM},
	doi = {10.1145/285861}
}

@article{AsanoKK2001,
	author    = {Tetsuo Asano and Naoki Katoh and Kazuhiro Kawashima},
	title     = {A New Approximation Algorithm for the Capacitated Vehicle Routing Problem on a Tree},
	journal   = {Journal of Combinatorial Optimization},
	volume    = {5},
	number    = {2},
	pages     = {213--231},
	year      = {2001},
	doi       = {10.1023/A:1011461300596},
	url       = {https://doi.org/10.1023/A%3A1011461300596}
}

@article{AdamaszekCL2010,
	author = {Anna Adamaszek and Artur Czumaj and Andrzej Lingas},
	title = {{PTAS} for $k$-tour cover problem on the plane for moderately large values of $k$},
	journal = {International Journal of Foundations of Computer Science},
	volume = {21},
	number = {06},
	pages = {893-904},
	year = {2010},
	doi = {10.1142/S0129054110007623},
	URL = {https://doi.org/10.1142/S0129054110007623}
}

@article{DasMathieu2014,
	author = {Das, Aparna and Mathieu, Claire},
	title = {A Quasipolynomial Time Approximation Scheme for {E}uclidean Capacitated Vehicle Routing},
	journal = {Algorithmica},
	volume = {73},
	pages = {115-142},
	year = {2014},
	doi = {10.1007/s00453-014-9906-4},
	URL = {https://doi.org/10.1007/s00453-014-9906-4}
}

@InProceedings{Becker18,
	author = {Becker, Amariah},
	title = {A Tight $4/3$ Approximation for Capacitated Vehicle Routing in Trees},
	booktitle = {Approximation, Randomization, and Combinatorial Optimization. Algorithms and Techniques (APPROX/RANDOM 2018)},
	pages = {3:1--3:15},
	year = {2018},
	volume = {116},
	address = {Dagstuhl, Germany},
	URL = {https://drops.dagstuhl.de/entities/document/10.4230/LIPIcs.APPROX-RANDOM.2018.3},
	URN = {urn:nbn:de:0030-drops-94075},
	doi = {10.4230/LIPIcs.APPROX-RANDOM.2018.3}
}

@article{GrandoniMZ2022unsplittablR2,
	title = {Unsplittable {E}uclidean Capacitated Vehicle Routing: A $(2+\epsilon)$-Approximation Algorithm},
	author = {Fabrizio Grandoni and Claire Mathieu and Hang Zhou},
	year = {2022},
	volume = {2209.05520},
	journal = {arXiv},
	primaryClass = {cs.DS},
	url = {https://arxiv.org/abs/2209.05520},
	doi = {10.48550/arXiv.2209.05520}
}

@article{blauth2023improving,
	title = {Improving the approximation ratio for capacitated vehicle routing},
	author = {Blauth, Jannis and Traub, Vera and Vygen, Jens},
	journal = {Mathematical Programming},
	volume = {197},
	number = {2},
	pages = {451--497},
	year = {2023},
	publisher = {Springer},
	doi = {10.1007/s10107-023-01858-5}
}

@article{mathieu2023ptastrees,
	title = {A {PTAS} for capacitated vehicle routing on trees},
	author = {Mathieu, Claire and Zhou, Hang},
	journal = {ACM Transactions on Algorithms},
	volume = {19},
	number = {2},
	pages = {1--28},
	year = {2023},
	publisher = {ACM New York, NY}
}

@article{JayaprakashS2023,
	author = {Jayaprakash, Aditya and Salavatipour, Mohammad R.},
	title = {Approximation Schemes for Capacitated Vehicle Routing on Graphs of Bounded Treewidth, Bounded Doubling, or Highway Dimension},
	year = {2023},
	issue_date = {April 2023},
	publisher = {Association for Computing Machinery},
	address = {New York, NY, USA},
	volume = {19},
	number = {2},
	issn = {1549-6325},
	url = {https://doi-org.vu-nl.idm.oclc.org/10.1145/3582500},
	doi = {10.1145/3582500},
	journal = {ACM Transactions on Algorithms},
	month = mar,
	articleno = {20},
	numpages = {36},
	pages = {1-36}
}

@article{MathieuZhou2024Tight-1.5trees,
	title = {A tight -approximation for unsplittable capacitated vehicle routing on trees},
	author = {Mathieu, C. and Zhou, H},
	journal = {Mathematical Programming},
	volume = {212},
	pages = {115-146},
	year = {2024}
}

@inproceedings{MiltenburOS2024,
	author = {Miltenburg, Steven and Oosterwijk, Tim and Sitters, Ren\'{e}},
	title = {Complexity of Fixed Order Routing},
	year = {2025},
	booktitle = {Approximation and Online Algorithms: 22nd International Workshop, WAOA 2024, Egham, UK, September 5–6, 2024, Proceedings},
	pages = {183–197},
	numpages = {15}
}

@article{Mitchell1999,
	title = {Guillotine subdivisions approximate polygonal subdivisions: A simple polynomial-time approximation scheme for geometric {TSP}, $k$-{MST}, and related problems.},
	author = {Joseph S. B. Mitchell},
	journal = {SIAM Journal on Computing},
	volume = {28},
	pages = {1298-1309},
	year = {1999},
	doi = {10.1137/0228009X}
}

@article{labbe1991capacitated,
	title = {Capacitated vehicle routing on trees},
	author = {Labb{\'e}, Martine and Laporte, Gilbert and Mercure, H{\'e}lene},
	journal = {Operations Research},
	volume = {39},
	number = {4},
	pages = {616--622},
	year = {1991},
	publisher = {INFORMS},
	doi = {10.1287/opre.39.4.616}
}

@article{DR59,
	author = {Dantzig, George Bernard and Ramser, John Hubert},
	title = {The Truck Dispatching Problem},
	journal = {Management Science},
	volume = {6},
	number = {1},
	pages = {80--91},
	year = {1959},
	publisher = {INFORMS},
	doi = {10.1287/mnsc.6.1.80}
}

@InProceedings{HamaguchiK1998,
	author = {Hamaguchi, Shin-ya and Katoh, Naoki},
	editor = {Chwa, Kyung-Yong and Ibarra, Oscar H.},
	title = {A Capacitated Vehicle Routing Problem on a Tree},
	booktitle = {Algorithms and Computation},
	year = {1998},
	publisher = {Springer Berlin Heidelberg},
	address = {Berlin, Heidelberg},
	pages = {399--407},
	isbn = {978-3-540-49381-5}
}

@article{Sitters2021,
	author = {Sitters, Ren\'{e}},
	title = {Polynomial Time Approximation Schemes for the Traveling Repairman and Other Minimum Latency Problems},
	journal = {SIAM Journal on Computing},
	volume = {50},
	number = {5},
	pages = {1580-1602},
	year = {2021},
	doi = {10.1137/19M126918X}
}

@inproceedings{GriesbachHKS_ESA_2023,
	author = {Svenja M. Griesbach and Felix Hommelsheim and Max Klimm and Kevin Schewior},
	title = {Improved Approximation Algorithms for the Expanding Search Problem},
	booktitle = {31st Annual European Symposium on Algorithms, {ESA} 2023, September 4-6, 2023, Amsterdam, The Netherlands},
	series = {LIPIcs},
	volume = {274},
	pages = {54:1--54:15},
	year = {2023},
	doi = {10.4230/LIPICS.ESA.2023.54}
}

@InProceedings{MiltenburgSOFSEM,
	author = {Miltenburg, Steven and Oosterwijk, Tim and Sitters, Ren{\'e}},
	title = {On the Complexity of Capacitated Vehicle Routing with Order Restrictions},
	booktitle = {SOFSEM 2026: Theory and Practice of Computer Science},
	series = {LNCS},
	year = {2026},
	publisher = {Springer},
	pages = {260--273},
}
	
		\end{document}